\documentclass[a4paper,onecolumn,11pt,accepted=2023-10-03]{quantumarticle}
\pdfoutput=1

\usepackage[utf8]{inputenc}
\usepackage[english]{babel}
\usepackage[T1]{fontenc}
\usepackage{amsmath}
\usepackage{hyperref}

\usepackage{tikz}
\usepackage{lipsum}
\usepackage[numbers,sort&compress]{natbib}

\usepackage{amsfonts}
\usepackage{amsmath}
\usepackage{amsthm}
\usepackage{amssymb}
\usepackage{bbm}
\usepackage{soul}

\usepackage{graphicx}
\usepackage{subcaption}
\usepackage{lipsum}

\usepackage{url}
\usepackage{booktabs}
\usepackage{nicefrac}
\usepackage{microtype}

\usepackage{algorithm}
\usepackage{algorithmic}

\usepackage{array, makecell}
\usepackage{booktabs}
\usepackage{multirow}

\newcommand{\EE}{\mathrm{E}}

\newcommand{\braket}[2]{\langle #1|#2\rangle}
\newcommand{\bra}[1]{\langle #1|}
\newcommand{\ket}[1]{|#1\rangle}
\newcommand{\ketbra}[1]{|#1\rangle\langle#1|}

\newtheorem{theorem}{Theorem}
\newtheorem{definition}{Definition}
\newtheorem{proposition}{Proposition}

\begin{document}

\title{Quantum Deep Hedging}

\author{El Amine Cherrat}
\email{el.amine-cherrat@qcware.com}
\affiliation{QC Ware}
\affiliation{Université de Paris, CNRS, IRIF}
\orcid{0009-0006-8682-3287}

\author{Snehal Raj}
\affiliation{QC Ware}
\orcid{0000-0002-4063-3190}

\author{Iordanis Kerenidis}
\affiliation{QC Ware}
\affiliation{Université de Paris, CNRS, IRIF}
\orcid{0000-0003-0659-3727}

\author{Abhishek Shekhar}
\affiliation{Quantitative Research, JPMorgan Chase}

\author{Ben Wood}
\affiliation{Quantitative Research, JPMorgan Chase}

\author{Jon Dee}
\affiliation{Quantitative Research, JPMorgan Chase}

\author{Shouvanik Chakrabarti}
\affiliation{Global Technology Applied Research, JPMorgan Chase}
\orcid{0000-0002-9159-4881}

\author{Richard Chen}
\affiliation{Global Technology Applied Research, JPMorgan Chase}
\orcid{0000-0002-5912-5620}

\author{Dylan Herman}
\affiliation{Global Technology Applied Research, JPMorgan Chase}
\orcid{0000-0002-8721-7848}

\author{Shaohan Hu}
\affiliation{Global Technology Applied Research, JPMorgan Chase}
\orcid{0000-0002-2877-2665}

\author{Pierre Minssen}
\affiliation{Global Technology Applied Research, JPMorgan Chase}
\orcid{0000-0003-0361-6962}

\author{Ruslan Shaydulin}
\affiliation{Global Technology Applied Research, JPMorgan Chase}
\orcid{0000-0002-8657-2848}

\author{Yue Sun}
\email{yue.sun@jpmorgan.com}
\affiliation{Global Technology Applied Research, JPMorgan Chase}
\orcid{0000-0002-0756-164X}

\author{Romina Yalovetzky}
\affiliation{Global Technology Applied Research, JPMorgan Chase}
\orcid{0000-0001-8397-2072}

\author{Marco Pistoia}
\affiliation{Global Technology Applied Research, JPMorgan Chase}
\orcid{0000-0001-9002-1128}

\maketitle

\begin{abstract}
    Quantum machine learning has the potential for a transformative impact across industry sectors and in particular in finance. In our work we look at the problem of hedging where deep reinforcement learning offers a powerful framework for real markets. We develop quantum reinforcement learning methods based on policy-search and distributional actor-critic algorithms that use quantum neural network architectures with orthogonal and compound layers for the policy and value functions. We prove that the quantum neural networks we use are trainable, and we perform extensive simulations that show that quantum models can reduce the number of trainable parameters while achieving comparable performance and that the distributional approach obtains better performance than other standard approaches, both classical and quantum. We successfully implement the proposed models on a trapped-ion quantum processor, utilizing circuits with up to $16$ qubits, and observe performance that agrees well with noiseless simulation. Our quantum techniques are general and can be applied to other reinforcement learning problems beyond hedging.
\end{abstract}

\section{Introduction}

In financial markets, hedging is the important activity of trading with the aim of reducing risk. For example, buyers and sellers of derivative contracts will often trade the asset underlying the derivative in order to mitigate the risk of adverse price movements. Classical financial mathematics provides optimal hedging strategies for derivatives in idealized friction-less markets, but for real markets these strategies must be adapted to take into account transaction costs, market impact, limited liquidity, and other constraints. 

Finding optimal hedging strategies in the presence of these important real-world effects is highly challenging. Deep Hedging~\cite{buehler_deep_2019, buehler_deep_2019-1} is a framework for the application of modern reinforcement learning techniques to solve this problem. One starts by defining a reinforcement learning environment for the hedging problem and a trading goal of maximizing a risk-adjusted measure of cumulative future returns. Then, one can apply standard deep reinforcement learning algorithms, such as  policy-search or actor-critic approaches, by designing neural network architectures to model the trading strategy and by defining a training loss function to find the optimal parameters that maximize the trading goal.

Beyond Deep Hedging, the applicability of machine learning to finance has grown significantly in recent years as highly efficient machine learning algorithms have evolved over time to support different data types and scale to larger data sets. For instance, supervised learning can be used for asset pricing or portfolio optimization~\cite{gu_empirical_2018, choi_stock_2018}, unsupervised learning for portfolio risk analysis and stock selection~\cite{zhu_pagan_2020, zhang_stock_2019}, and reinforcement learning for algorithmic trading~\cite{cartea_deep_2021, deng_deep_2017}. At the same time, machine learning has been identified as one of the most important domains of applicability of quantum computing given the potential ability of quantum computers to solve classically-intractable computational problems~\cite{liu_rigorous_2021}, perform linear algebraic operations efficiently~\cite{chakraborty_power_2019}, compute gradients~\cite{gilyn_optimizing_2019} and provide variational-type approaches~\cite{cerezo_variational_2021}. 
Such techniques have already been considered for financial use cases~\cite{kerenidis_quantum_2019, leclerc_financial_2022, emmanoulopoulos_quantum_2022, rebentrost_quantum_2022, doriguello_quantum_2022,Niroula2022}, and in fact finance is estimated to be one of the first industry sectors to benefit from quantum computing~\cite{menard2020game,herman2022}.

 In this work, we develop quantum deep learning methods and show how they can be a powerful tool for Deep Hedging. Quantum deep learning methods, and, in particular, quantum neural networks based on parametrized quantum circuits, have been proposed as a way to enhance the power of classical deep learning~\cite{cerezo_variational_2021}. Such quantum neural networks may in general be difficult to train, often encountering problems of barren plateaus or vanishing gradients~\cite{mcclean_barren_2018}. Here, we design quantum neural layers based on Hamming-weight preserving unitaries constructed out of 2-dimensional rotation gates (called RBS gates), and prove that they can be efficiently trainable, in the sense that the variance of the gradients decays only polynomially with the number of qubits. These quantum layers, first defined in~\cite{kerenidis_classical_2022}, naturally provide models that have orthogonal features, improving interpretability~\cite{yang_enhancing_2021}, allowing for deeper architectures, and providing theoretical and practical benefits in  generalization~\cite{li_orthogonal_2021}. Depending on the input data encoding, one can control the size of the Hilbert space these neural networks are exploring while training. We call these two types of layers \emph{orthogonal layers} and \emph{compound layers} respectively. 
 
 First, using our orthogonal layers within classical neural network architectures, we design novel quantum neural networks for time-series. To evaluate the behavior of our quantum neural networks, we use the same example as in~\cite{buehler_deep_2019}, where the market was simulated using Geometric Brownian Motion (GBM) with a single hedging instrument (equity). We then benchmark four different neural network architectures (Feed-forward, Recurrent, LSTM, Transformer), with both classical layers and our quantum layers, using the policy-search Deep Hedging algorithm~\cite{buehler_deep_2019, buehler_deep_2019-1}. Our quantum neural networks achieve comparable scores as their classical counterparts while obtaining qualitatively different solutions, providing models with orthogonal features and considerably fewer parameters. It is conceivable that similar parameter reduction may be obtained by purely classical techniques such as pruning.
 
 Second, we design a novel quantum-native reinforcement learning method for Deep Hedging. We start by formulating a quantum encoding of the environment and trading goals for Deep Hedging. We then introduce a distributional actor-critic reinforcement learning algorithm in combination with quantum neural network architectures based on compound layers for the policy and value function. Our approach is inspired by classical distributional reinforcement learning wherein the critic does not only learn the expectation of cumulative returns, but also approximates their distribution. Recent studies, such as AlphaTensor~\cite{fawzi_discovering_2022}, have demonstrated that distributional reinforcement learning can lead to better models compared to standard reinforcement learning methods, despite the increased difficulty in training~\cite{lyle_comparative_2019}. A distributional actor-critic method for Deep Hedging had not been studied in the classical case before. We note that our quantum distributional reinforcement learning algorithms can be used more generally for reinforcement learning problems and is not limited to Deep Hedging. 

Quantum computers are naturally suited to distributional reinforcement learning. Each quantum circuit explicitly encodes a mapping between exponential size distributions, and the measurement of a quantum circuit results in a sample from such a distribution. These samples can be used to simply learn the expectation of the entire distribution or to flexibly obtain extra information about the distribution, such as the expectations restricted to relevant subsets of the total range. In the case of Deep Hedging, we parametrize the value function using quantum neural networks with compound layers, which preserve the Hamming-weight subspaces of their input domain. This choice is particularly suited to the Deep Hedging setting, where when we encode a stochastic path as a binary string of up or down jumps, then it is intuitively natural that the number of net \emph{jumps}, namely the Hamming weight of the encoding, is a major component that determines its behavior. The restriction to compound layers further makes the neural network architecture shallower and trainable.  

Confirming this intuition, the quantum policies trained using our distributional actor-critic algorithm outperform those trained with policy-search based or standard actor-critic models where only the expectation of the value function is learned. These results are achieved again for a variant of the example from~\cite{buehler_deep_2019}, where we used a discretized Geometric Brownian Motion as a market model both without and with transaction costs. 

Last, we evaluated our framework on
the Quantinuum H1-1 and H1-2 trapped-ion quantum processors~\cite{h1}. In particular, we performed inference on the quantum hardware using two sets of Quantum Deep Hedging models which were classically pre-trained. First, we used the policy-search based algorithm with the LSTM and Transformer architectures instantiated with 16-qubit orthogonal layers. Second, we used the novel distributional actor-critic algorithm instantiated with compound neural networks using up to 12 qubits. We observed close alignment between noiseless simulation and hardware experiments, with our distributional actor-critic models again providing best performance. {We note that some of the circuits used to instantiate our framework may be classically simulatable with only a polynomial overhead (including the setup used in our numerical experiments)~\cite{Brod2016,2308.01432}. Nevertheless, this does not hold true for the Quantum Deep Hedging framework, which is more general and can be applied with any quantum layers. For example, it can be shown that with suitable input states which are still very easy to create, circuits used in our work become classically hard to simulate ~\cite{oszmaniec_fermion_2022}}.

The rest of the paper is organized as follows. Section~\ref{sec:preliminaries} introduces preliminaries for quantum computing and reinforcement learning. In Section~\ref{sec:deep-hedging}, the problem of Deep Hedging is formulated and policy-search and actor-critic algorithms are presented. Section~\ref{sec:quantum-NNs} presents our orthogonal and compound neural networks and proves their trainability. Section~\ref{sec:quantum-deep-hedging} introduces a novel quantum framework for reinforcement learning and applies it to Deep Hedging. Section~\ref{sec:applications} reports our simulation and hardware implementation results. Finally, Section~\ref{sec:conclusions} concludes with remarks and open questions.

\section{Preliminaries}
\label{sec:preliminaries}

\subsection{Quantum Computing}
\label{subsec:quantum-computing}

Quantum computing~\cite{nielsen_quantum_2012} is a new paradigm for computing that uses the postulates of quantum mechanics to perform computation. The basic unit of information in quantum computing is a qubit. The state of a qubit can be written as: 
$$\ket{x} = \alpha_0 \ket{0} + \alpha_1\ket{1},$$ 
where $\alpha_0,\alpha_1 \in \mathbb{C}$ and $|\alpha_0|^2+|\alpha_1|^2=1$, and corresponds to a unit vector in the Hilbert space $\mathcal{H}=\text{span}_{\mathbb{C}}\{\ket{0},\ket{1}\}$. 
A qubit can be generalized to $n$-qubit states, which are represented by unit vectors in $\mathcal{H}^{\otimes n}\simeq\mathbb{C}^{2^n}$. 
Specifically,
$$\ket{x} = \sum_{b\in\{0,1\}^n} \alpha_b \ket{b},$$ 
with $\alpha_b=\braket{x}{b} \in \mathbb{C}$ and $\sum_{b\in\{0,1\}^n} |\alpha_b|^2=1$, where $\{\ket{b}\:|\: b\in\{0,1\}^n\}$ is the computational basis of $\mathcal{H}^{\otimes n}$. Quantum states evolve through the application of unitary operators, which are $2^n \times 2^n$ unitary matrices. Applying a unitary operator $U$ to an $n$-qubit state $\ket{x}$ results in a new quantum state $U\ket{x}$.

Quantum states can be measured, and the measurement process reveals information about the state of the system. The probability of a quantum state $\ket{x}$ yielding outcome $b$ from a measurement in the computational basis is given by $|\alpha_b|^2$.
More generally, the measurement process is described as an observable, which is a Hermitian operator that acts on the quantum state. The observable is given by $O = \sum_m \text{o}_m P_m$, where $\text{o}_m$ are real numbers that specify the measurement outcomes and $P_m$ are projection operators onto the subspaces that correspond to each outcome.  This can be calculated as the inner product between the state and the corresponding projection operator, or $p_m(x) = \braket{x}{P_m|x}$. The expectation of measuring the observable $O$ in the state $\ket{x}$ is defined as the sum of the measurement outcomes weighted by their corresponding probabilities, or $\sum_m \text{o}_m p_m(x)$. This expectation can also be written as the trace of the observable $O$ and the density matrix $\rho(x) = \ket{x}\bra{x}$, i.e., $\braket{x}{O|x} = \mathrm{Tr}\left[O\rho(x)\right]$, where $\mathrm{Tr}$ is the trace operator. 
The trace operation returns the sum of the diagonal elements of the matrix, which corresponds to the expected value of the measurement outcome in the state $\ket{x}$.
 
\subsection{Reinforcement Learning}
\label{subsec:reinforcement-learning}

The aim of reinforcement learning~\cite{sutton_reinforcement_1998} is to train an agent to discover the policy that maximizes the agent's performance in terms of cumulative future reward.
And while interacting with the environment, the agent only receives a reward signal. 
The agent can take an action from a set of possible actions based on a policy that maps each state to an action.

Environments in reinforcement learning are modeled as decision making problems defined by specifying the state set, the action set, the underlying model describing the dynamics of the environment, and the reward mechanism.
The usual framework used to describe the environment's elements in reinforcement learning are \emph{Markov Decision Processes} (MDPs). In this paper, we will consider finite-horizon MDPs that can be defined as follows:

\begin{definition}[Finite-horizon MDP]
A finite-horizon MDP $\mathcal{M}$ is defined by 
a tuple $(\mathcal{S},\mathcal{A},p,r,T)$, where $\mathcal{S}$ is is the state space, $\mathcal{A}$ is the action space, $p:\mathcal{S}\times\mathcal{A}\xrightarrow[]{}\Delta(\mathcal{S})$ is the transition function with $\Delta(\mathcal{S})$ the set of distributions over $\mathcal{S}$, $r:\mathcal{S}\times\mathcal{A}\xrightarrow[]{}\mathbb{R}$ is the reward function and $T\in\mathbb{N}^*$ is the time horizon.
\end{definition}

Starting from a state $s_t \in \mathcal{S}$, a single interaction with the environment can be represented by a sequence of actions $\{a^\pi_{t'}\}_{t'=t}^{T}$ selected based on a deterministic policy $\pi: \mathcal{S} \to \mathcal{A}$, and a sequence of random states $\{s_{t'}\}_{t'=t}^{T}$ that follow the MDP transitions $p$. The cumulative return $R^\pi_t$ is the sum of rewards from time-step $t$ to $T$ and is given by
$$ R^\pi_t(s_t, s_{t+1}, \dots, s_T) = \sum_{t'=t}^T r(s_{t'},a_{t'}^\pi). $$ 

In reinforcement learning, the objective is to find the policy $\pi^*$ that maximizes the expected return for all states $s_t$. The expected value of the return, taking into consideration all possible future states $\{s_{t'}\}_{t'=t}^{T}$ resulting from the environment transitions described by $p$, is referred to as the value function $v^\pi$. For any time-step $t$, and denoting by $\boldsymbol{s}_{t'}\in\Delta(\mathcal{S})$ the random variable that takes values in $\mathcal{S}$ according to the environment dynamics for $t'> t$ and knowing $\boldsymbol{s}_t=s_t$, the cumulative return is defined as follows: $$ v_t^\pi(s_t) = \mathbb{E}[R_t^\pi(\boldsymbol{s}_t,\boldsymbol{s}_{t+1},\dots,\boldsymbol{s}_{T}) | s_t]. $$ 

Typically, the goal of reinforcement learning is to find a policy that maximizes the expected return, and different algorithms have been developed to achieve such objectives~\cite{arulkumaran_deep_2017}. The value function is used to evaluate policies in order to find the one that maximizes the expected return.

\section{Deep Hedging Formulation}
\label{sec:deep-hedging}

\subsection{Deep Hedging Environment}
\label{subsec:deep-hedging-environment}

Deep Hedging is a classical algorithm that treats hedging of a set of derivatives as a reinforcement learning problem. This algorithm was first introduced in~\cite{buehler_deep_2019,buehler_deep_2019-1} and has been further developed in subsequent works such as~\cite{wiese_deep_2019, buehler_deep_2022, wiese_multi-asset_2021, murray_deep_2022}. In the original approach, the authors use a reinforcement learning environment associated with Deep Hedging that employs finite-horizon MDPs where there is a different state and action space per time-step. The time horizon $T$ represents the maximum maturity of all instruments, and $\mathcal{S}_t$ and $\mathcal{A}_t$ are the sets of observed market states and available actions at time-step $t$, respectively.

During the interaction with the environment, the agent observes a market state $s_t\in\mathcal{S}_t$ that contains all current and past market information (prices, cost estimates, news, internal state of neural networks, $\dots$), and takes an action $$a_t^\pi = \pi_t(s_t) \in \mathcal{A}_t,$$ potentially subject to constraints (liquidity limits, risk limits, $\dots$), according to a deterministic policy $\pi:=\{\pi_t\}_{t=0}^T$. 
Then, the environment transitions to the next state $s_{t+1}$, according to $p_t:\mathcal{S}_t \xrightarrow[]{} \Delta(\mathcal{S}_{t+1})$, that is assumed not to depend on the action $a_t^\pi$ since actions have no market impact in the Deep Hedging model~\cite{murray_deep_2022}:
$$p_t(s_{t+1}|s_t):= \mathbb{P}[s_{t+1}|s_t].$$
Subsequently, the agent receives a total reward of
$$r_t^\pi(s_t):=r_t(s_t,a^\pi_t) = r^+_t - r^-_t,$$ 
where $r^+_t$ is the source of positive rewards such as the generated cashflows and $r^-_t$ represents the source of negative rewards and corresponds to the transaction costs. 
The interaction ends after a terminal state $s_T\in\mathcal{S}_T$ is reached.
The cumulative sum of the rewards perceived during this interaction can be rewritten as 
$$R_t^\pi(s_T)\equiv R_t^\pi(s_t,s_{t+1},\dots, s_T),$$ 
since, by definition, $s_T$ contains all previous states $s_{t'}$ for all $t \leq t'<T$. 

\subsection{Trading Goals in Deep Hedging}
\label{subsec:deep-hedging-trading-goal}

The standard objective in reinforcement learning problems is to find the optimal strategy $\pi^*$ that maximizes the value function $v^\pi$ over all policies $\pi$.
As discussed in Section~\ref{subsec:reinforcement-learning}, the value function is usually defined as the expected cumulative return, which, in this context, would be $\mathbb{E}[R^\pi_t(\boldsymbol{s}_T)|s_t]$. However, in order to take into account the inherent risk in trading strategies, the goal of Deep Hedging is to find a deterministic optimal policy $\pi^*$ that maximizes the value function defined, for some policy $\pi$, as
\[
    v_t^\pi(s_t) = \EE \left[ R_t^\pi(\boldsymbol{s}_T)|s_t \right],
\]
where $\EE$ is the expected utility defined by a risk-averse (i.e., concave) utility function. 

Various forms of utility functions have been proposed that satisfy the concave requirement.
For a more detailed discussion on the desired properties of the utility function and examples of commonly used forms, see~\cite{buehler_deep_2019, buehler_deep_2019-1}. 
In this paper, we use the exponential utility function $\EE_\lambda$, which is an example of a monetary utility function that is increasing, concave, and cash-invariant~\cite{murray_deep_2022}. Specifically, it is defined for some risk aversion level $\lambda > 0$ as
$$ \EE_\lambda[\boldsymbol{X}] := -\frac{1}{\lambda} \log \mathbb{E} [\exp(-\lambda \boldsymbol{X})]. $$
The parameter $\lambda$ can be used to reflect the investor's risk tolerance, with larger values indicating more risk aversion. With this exponential utility function, the value function $v^*$ associated with the optimal policy $\pi^*$ is given by $$ \forall t, \: \forall s_t \in \mathcal{S}_t, \quad v^*_t(s_t) = \sup_{\pi} v_t^\pi(s_t) = -\frac{1}{\lambda}\log \left\{\inf_{\pi}  \mathbb{E}\left[\exp \left(-\lambda R_t^\pi\left((\boldsymbol{s}_T\right) \right)|s_t\right]\right\}. $$

Since the Deep Hedging objective is formulated in terms of risk-adjusted measures, the value function is no longer a solution to the standard Bellman equation. However, conventional reinforcement learning algorithms can be adapted to find policies that maximize the utility. Two algorithms have been developed to solve the Deep Hedging problem. The first approach, called policy-search Deep Hedging~\cite{buehler_deep_2019, buehler_deep_2019-1}, uses a neural network to model the policy and updates the set of parameters using gradient descent to minimize the policy loss function. The second approach, actor-critic Deep Hedging~\cite{murray_deep_2022}, represents both the policy and the value function with neural networks and computes the utility using the policy function to update the value network, which is then used to update the policy network.

\section{Quantum Neural Networks with Orthogonal and Compound Layers}
\label{sec:quantum-NNs}

A Quantum Neural Network (QNN) consists of a composition of parametrized unitary operations, whose parameters can be trained to provide machine learning models for classification or regression tasks. While current quantum hardware is still far from being powerful enough to compete with classical machine learning algorithms, many interesting quantum machine learning algorithms have started to appear, such as for regression~\cite{mitarai_quantum_2018, herman_expressivity_2022}, classification~\cite{farhi_classification_2020, perez-salinas_data_2020, landman_quantum_2022}, generative modeling~\cite{benedetti_generative_2019, benedetti_variational_2021}, and reinforcement learning~\cite{meyer_survey_2022}.

In general, a QNN consists of data-loading layers and trainable layers, which are both parametrized unitary operations. In some architectures, the data-loading is an explicit-encoding scheme that is used to directly embed the classical input data into the amplitudes of a quantum state. While in others, these parts only implicitly encode the data and prepare a state whose amplitudes are some complex non-linear function of the input data. The latter is known in literature as a quantum feature map~\cite{havlicek_supervised_2019}. After the unitary operators are applied, the resulting quantum state is probed to produce classical data that can be used for inference or training.

One popular approach, based on variational quantum circuits~\cite{cerezo_variational_2021}, is to apply an alternating sequence of quantum feature maps and trainable parts and output the expectation of the resulting state with respect to some observable~\cite{schuld_effect_2021, gil_vidal_input_2020}. A second class of architectures encodes in the amplitudes the input data and performs trainable unitary operations that reproduce the linear layers of certain classical neural networks but with reduced computational complexity~\cite{kerenidis_classical_2022, cherrat_quantum_2022}. The output is obtained through quantum-state tomography and non-linear operations are then applied classically. After applying the non-linearity, the data is reloaded onto a quantum state, and the process is repeated to compose layers. 

Such quantum circuits can be trained using classical gradient descent methods until convergence. For variational quantum circuits, where the output is an expectation value, the gradient can be computed using the parameter-shift rule~\cite{mitarai_quantum_2018, schuld_evaluating_2019}. One needs to be very careful in designing such quantum neural networks, since they may in general be difficult to train, often encountering problems, such as, barren plateaus or vanishing gradients~\cite{mcclean_barren_2018}. 

In Sections~\ref{subsec:quantum-analogues-of-classical-nns}~and~\ref{sec:compound_layer} below we review two different types of quantum layers built from Hamming-weight preserving unitaries and that can be used to provide a natural quantization of classical neural network architectures. 
In Section~\ref{subsec:qnn-archs-for-timeseries} we discuss neural networks architectures that make use of these layers. While similar techniques have appeared previously in~\cite{cherrat_quantum_2022, landman_quantum_2022}, our discussion is more systematic and extends the techniques to a larger set of architectures. 
Finally, in Section~\ref{subsec:quantum-compound-nns} we discuss the properties of quantum neural networks with compound layers and prove their trainability.

\subsection{Quantum Orthogonal Layers}
\label{subsec:quantum-analogues-of-classical-nns}

The quantum orthogonal layer was proposed by Kerenidis et al.~\cite{kerenidis_classical_2022} to simulate traditional orthogonal layers with reduced complexity at inference time. Specifically, a quantum orthogonal layer on $n$-qubits acts as an element of $\text{SO}(n)$ when restricted to the span of the computational basis states with Hamming-weight one, i.e. the unary basis. This is achieved by composing two-qubit Reconfigurable Beamsplitter (RBS) gates. An RBS gate acting on the $i$-th and $j$-th qubits implements a Givens rotation:
$$ \text{RBS}_{ij}(\theta)  = \begin{pmatrix} 1 & 0 & 0 & 0 \\
     0 & \cos(\theta) & \sin(\theta) & 0 \\ 
     0 & - \sin(\theta) & \cos(\theta) & 0 \\
      0 & 0 & 0 & 1 \end{pmatrix}.
$$

If the goal is to apply an orthogonal matrix to classical data in a vector $\boldsymbol{x} \in \mathbb{R}^{n}$, then $\boldsymbol{x}$ can be efficiently amplitude encoded in a quantum state in the unary basis with a log-depth circuit~\cite{kerenidis_method_2020, johri_nearest_2021}. The unary data loader (depicted in Figure~\ref{fig:parallel_loader}) uses $n$ qubits and maps the all-zeros basis state $\ket{0}^{\otimes n}$ to the state $|\boldsymbol{x}\rangle$ as follows:
$$ U_L(\boldsymbol{x}) : \ket{0}^{\otimes n} \rightarrow |\boldsymbol{x}\rangle = \frac{1}{\| \boldsymbol{x} \|} \sum^n_{i=1} x_i \ket{e_i},$$
where $\| \cdot \|$ represents the $\ell_2$ norm and $\ket{e_i}$ is the $i^{th}$ unary basis quantum state represented by $|0\rangle^{\otimes(i-1)}|1\rangle|0\rangle^{\otimes(n-i)}$. 

Let $G(i, j, \theta)$ denote the Givens rotation applied to the $i$-th and $j$-th unary basis vector, i.e. $e_i$ and $e_j$, $\boldsymbol{\theta}$ a vector of angles, and $\mathcal{T}$ is a list of triplets $(i, j, m)$. The orthogonal layer is defined by
$$U(\boldsymbol{\theta}) = \prod_{(i, j, m) \in \mathcal{T}}\text{RBS}_{ij}(\theta_{m}).$$ 
It acts as $ U(\boldsymbol{\theta})\ket{\boldsymbol{x}} = W\ket{\boldsymbol{x}}$, where $W =  \prod_{(i, j, m) \in \mathcal{T}}G(i, j, \theta_{m})$. Since the dimension of the Hamming-weight one subspace is $n$ for $n$ qubits, there exist efficient quantum-state tomography procedures for reading out the resulting quantum state encoding the matrix-vector product $W\ket{\boldsymbol{x}}$~\cite{kerenidis_classical_2022}.

The fact that circuits of $\text{RBS}$ can only span elements of $\text{SO}(n)$ avoids the computational overhead that is associated with the need to re-orthogonalize the weight matrix in the classical case. This can be implemented with a linear-depth quantum circuit. Note that the application of each such layer can also be performed classically in time $\mathcal{O}(n^2)$. Furthermore, orthogonal layers are efficiently trainable, as the dimension of the space they explore is linear in the number of qubits used. Specifically, these layers are trained by classically simulating the circuit and using quantum for inference.

\begin{figure}[t!]
    \centering
    \includegraphics[height=12em]{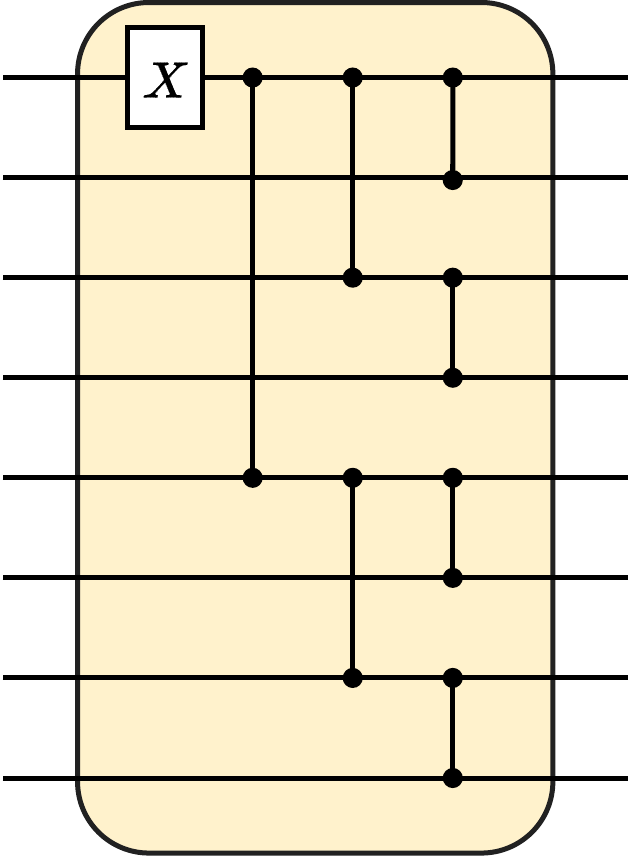}
    \caption{A quantum circuit with logarithmic depth for data loading. Vertical lines represent RBS gates with parameters that are dependent on the input $\boldsymbol{x}$. The unitary represented by this data loader is denoted as $U_{L}(\boldsymbol{x})$.}
    \label{fig:parallel_loader} 
\end{figure}

\begin{figure}[t!]
    \centering
    \begin{subfigure}{0.39\textwidth}
    \centering
    \includegraphics[height=12em]{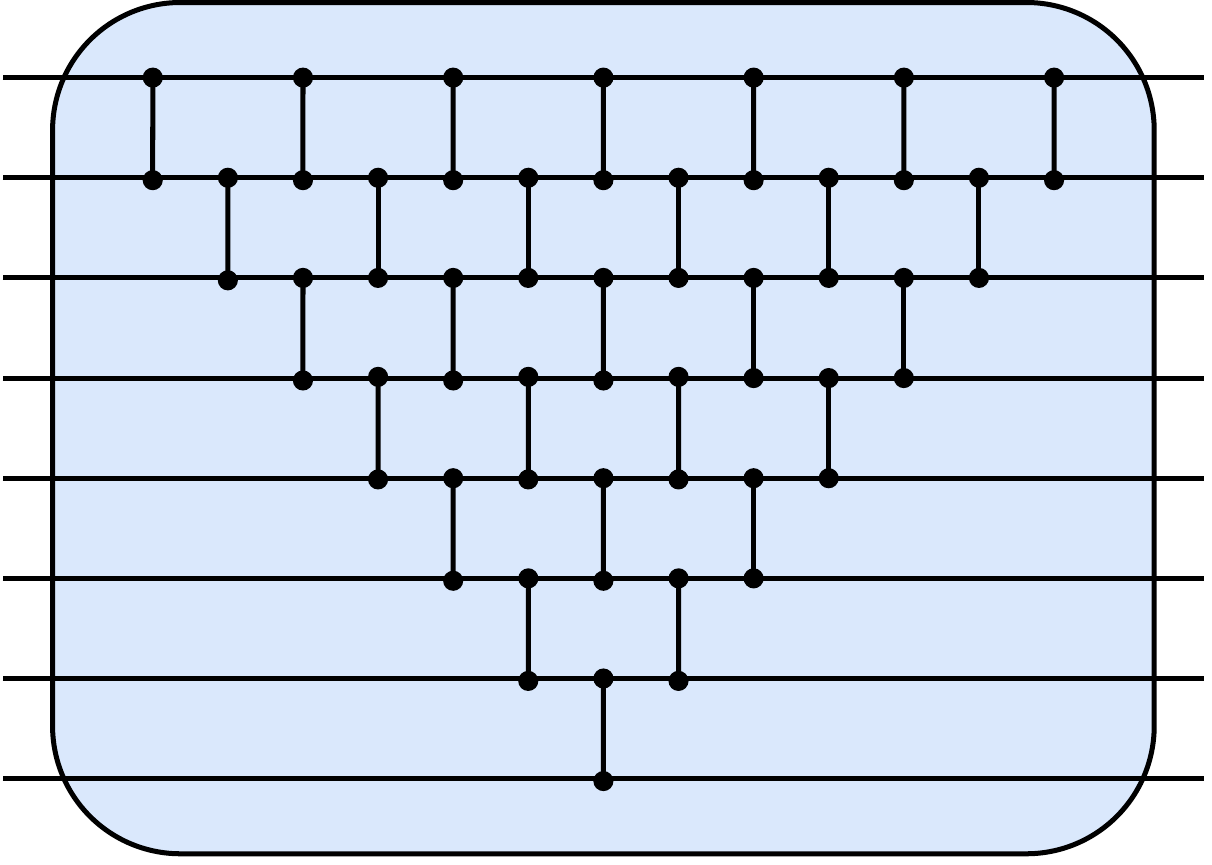} 
    \caption{Pyramid architecture}
    \end{subfigure}
    \begin{subfigure}{0.29\linewidth}
    \centering
    \includegraphics[height=12em]{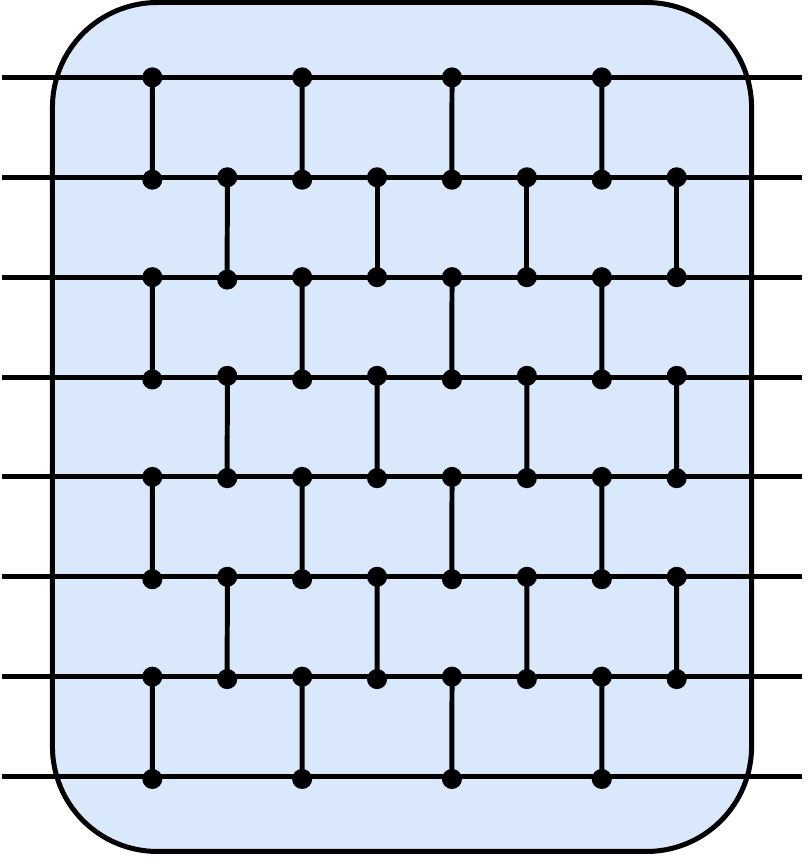}
    \caption{Brick architecture}
    \end{subfigure}
    \begin{subfigure}{0.29\linewidth}
    \centering
    \includegraphics[height=12em]{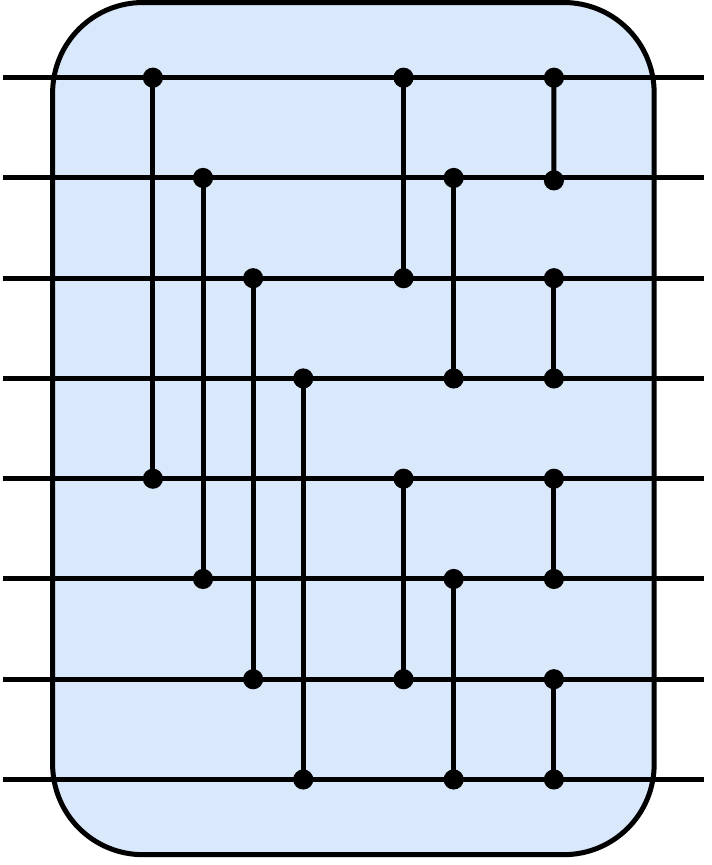} 
    \caption{Butterfly architecture}
    \end{subfigure}
    \caption{Various Hamming-weight preserving circuits used in quantum orthogonal layers. These circuits are parameterized by a set of parameters $\boldsymbol{\theta}$, with each parameter representing the angle of a specific RBS gate. The parameterized unitary represented by this layer is expressed as $U(\boldsymbol{\theta})$. }
    \label{fig:orthogonal-layers}
\end{figure}

There are different linear-depth circuits for $U(\boldsymbol{\theta})$, highlighted in Figure~\ref{fig:orthogonal-layers}, each with its own unique properties. The Pyramid architecture, as described in~\cite{kerenidis_classical_2022}, consists of $n(n-1)/2$ RBS gates arranged in a pyramid-like structure and has a linear depth. This architecture allows for the representation of all possible orthogonal matrices of size $n \times n$. The Brick architecture is a variation of the Pyramid architecture and also consists of $n(n-1)/2$ RBS gates. However, it has a more compact layout of gates, while still exhibiting similar properties. Both the Pyramid and Brick architectures can be implemented in hardware with nearest-neighbour connectivity between qubits.
\newpage
\begin{figure}[t!]
    \centering
  \includegraphics[width=0.75\textwidth]{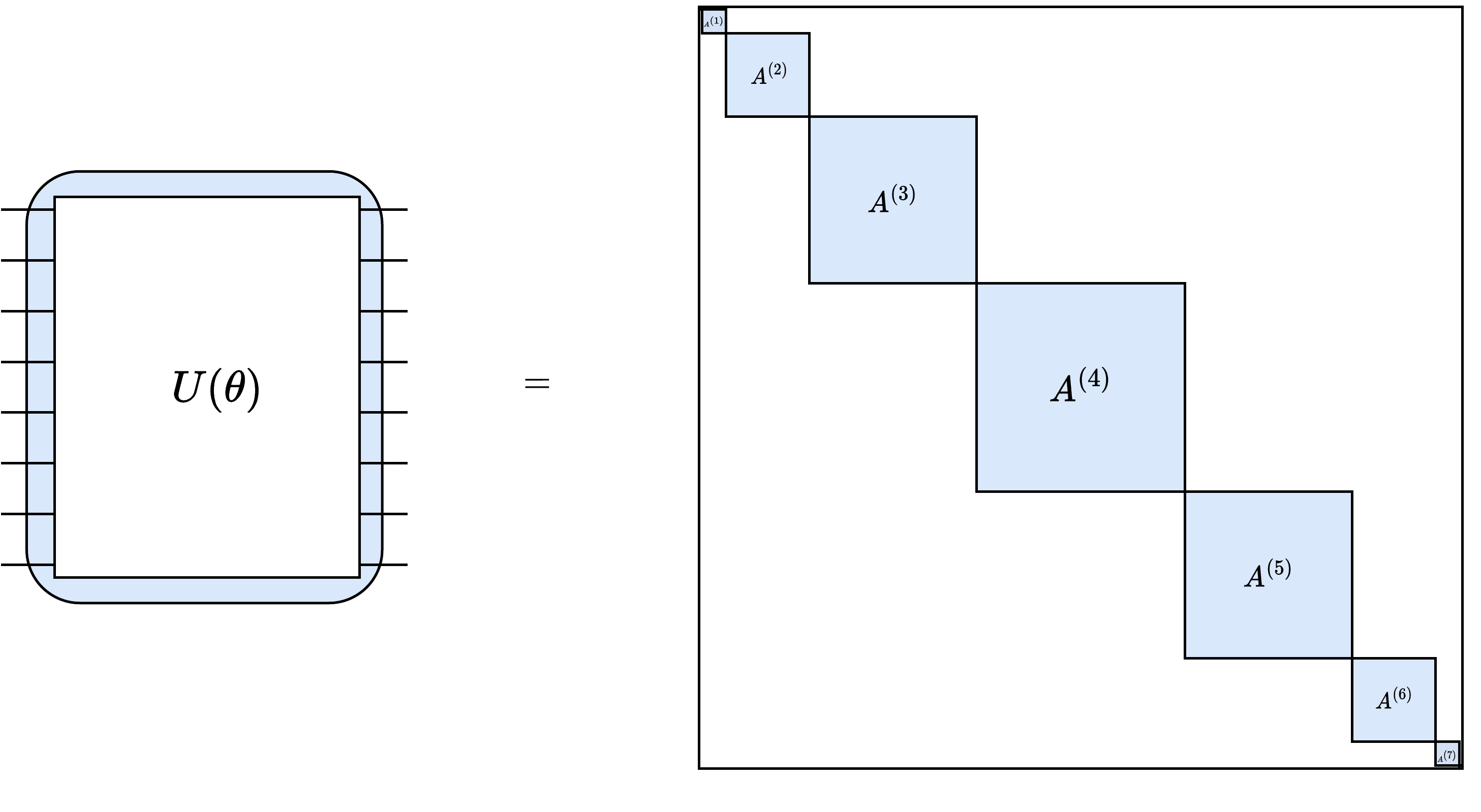}
  \caption{A quantum compound layer $U(\boldsymbol{\theta})$ acts as a block diagonal unitary on each fixed Hamming-weight subspace.}
  \label{block}
\end{figure}

On the other hand, the Butterfly architecture, which was proposed in~\cite{cherrat_quantum_2022}, uses logarithmic depth circuits with a linear number of gates to implement a quantum orthogonal layer. This architecture requires all-to-all connectivity in the hardware layout. To summarize, an orthogonal layer with input size $n$ uses a parametrized quantum circuit with $n$ qubits and a number of parameters equal to $\binom{n}{2}$ for a Pyramid or Brick circuit and can be implemented on hardware with nearest-neighbour connectivity. For a Butterfly circuit, the number of parameters is $\frac{n}{2} \log(n)$ and requires all-to-all connectivity in the hardware.

Because classical data can be efficiently loaded onto a quantum state and retrieved from a quantum orthogonal layer, it is possible to compose quantum orthogonal layers with nonlinear activation functions. More specifically, after applying the sequence of $\text{RBS}$ gates, the matrix-vector product $W\ket{\boldsymbol{x}}$ is readout and an activation function is applied classically. The resulting vector is then loaded onto a quantum state using the unary data loader for the next layer. In Section~\ref{subsec:qnn-archs-for-timeseries}, we will make use of this scheme to construct various quantizations of classical neural architectures for time series. 

\subsection{Quantum Compound Layers}
\label{sec:compound_layer}

The quantum compound layer is a natural and powerful generalization of the orthogonal layer~\cite{landman_quantum_2022} and a version of it has been previously used~\cite{cherrat_quantum_2022} to implement quantum analogues of vision transformers. The prefix ``compound'' refers to the fact that the quantum circuits implement linear operators on the exterior power of a vector space.

For an $n$-dimensional vector space $V$ with orthonormal basis $\{e_{i}\}_{i=1}^{n}$, the $k$-th exterior power $\bigwedge^{k} V$ is the $\binom{n}{k}$-dimensional vector space spanned by the $k$-fold alternating products of vectors in $V$:  $$\bigwedge\nolimits^{k} V:= \text{span
}_{\mathbb{C}}\{e_{i_1} \wedge e_{i_2} \wedge \dots e_{i_k} | i_1 < i_2 < \cdots < i_k \in [n]\}.$$ The alternating property implies that for any permutation $\sigma$ of  the indices: $$e_{\sigma(i_1)} \wedge e_{\sigma(i_2)} \wedge \cdots \wedge e_{\sigma(i_k)} = \text{sign}(\sigma)\times e_{i_1} \wedge e_{i_2} \wedge \cdots \wedge e_{i_k}.$$ The direct sum $ \bigoplus_{k=0}^{n}\bigwedge^{k} V$ equipped with the alternating product forms the exterior algebra of $V$, denoted $\bigwedge V$. 

For any linear operator $A$ on $V$, there exists an extension to a linear operator $A^{(k)}$ on $\bigwedge^{k} V$, which acts as
$$ A^{(k)}(e_{i_1} \wedge e_{i_2} \wedge \cdots e_{i_k}) = Ae_{i_1} \wedge Ae_{i_2} \wedge \cdots \wedge Ae_{i_k}$$
on $k$-vectors. The matrix for the extended operator, called the $k$-th (multiplicative) compound matrix, has as entries $A^{(k)}_{IJ} = \det(A_{IJ})$, where $I$ and $J$ are $k$-sized subsets of the rows and columns of $A$, respectively. Furthermore, there exists a unique linear operator $\mathcal{A} := \bigoplus_{k=0}^{n}A^{(k)}$ over $\bigwedge V$ such that the restriction to the $k$-th exterior power is $A^{(k)}$.

In the quantum setting, the $k$-vectors $e_{i_1} \wedge \cdots \wedge e_{i_k}$ are mapped to computational basis states $\ket{S}$, where $S \in \{0, 1\} ^{n}$, $|S| = k$, and $\forall t \in [k], S_{i_t} = 1$. Thus $n$-qubits can be used to encode a projectivization of the exterior algebra $\bigwedge V$. To apply compound matrices to $k$-vectors in the qubit encoding, we utilize Fermionic Beam Splitter (FBS) gates, whose action on qubits $i$ and $j$ depends on the parity of the qubits between $i$ and $j$. On a computational basis state $\ket{S}$, the FBS gate acts on qubits $i$ and $j$ as the unitary matrix below:
$$ \text{FBS}_{ij}(\theta) = \begin{pmatrix} 1 & 0 & 0 & 0 \\
     0 & \cos(\theta) & (-1)^{f(i, j, S)}\sin(\theta) & 0 \\ 
     0 & (-1)^{f(i, j, S)+1}\sin(\theta) & \cos(\theta) & 0 \\
      0 & 0 & 0 & 1 \end{pmatrix},
$$
where $\theta \in [0, 2\pi)$ and $f(i, j, S) = \sum_{i < k < j}s_k$, and as identity on all other qubits. An FBS gate can be implemented as the composition of controlled-$X$ gates, controlled-$Z$ gates, and RBS gates.

Like in the previous subsection, let $G(i, j, \theta)$ denote the Givens rotation applied to the $i$-th and $j$-th basis vector, i.e. $e_i$ and $e_j$, $\boldsymbol{\theta}$ a vector of angles, and $\mathcal{T}$ is a list of triplets $(i, j, m)$. The quantum compound layer is defined by
$$U(\boldsymbol{\theta}) = \prod_{(i, j, m) \in \mathcal{T}}\text{FBS}_{ij}(\theta_{m}).$$
It can be shown~\cite{kerenidis_quantum_2022} that this layer  acts as $ U(\boldsymbol{\theta})\ket{S} = A^{(k)}\ket{S} $ for $|S| = k$, where $A^{(k)}$ is the $k$-th multiplicative compound of  $A =  \prod_{(i, j, m) \in \mathcal{T}}G(i, j, \theta_{m})$. Thus the operation $U$ acts as $\mathcal{A}$ over the quantum state space, in other words it is a block diagonal unitary that acts separately on each fixed Hamming-weight subspace (see Figure~\ref{block}).

The compound layer is similar to the circuits we described in the orthogonal layer case, where there, given we only consider the unary basis, the FBS gates can be replaced by the simpler RBS gates. Note as well, that RBS and FBS gates acting on nearest-neighbor qubits, as in the Pyramid and Brick circuits, are also equivalent. The main difference in the compound layer comes from the fact the data loading part is not restricted to the unary basis, and thus one needs to consider the entire exponential size block-diagonal unitary, and not only its linear size restriction to the unary basis. 

Thus, one can see that by controlling the Hamming-weight of the basis states used in the data loading part one can smoothly control the size of the explored space, from linear size, when using a unary basis for data loading, to an exponential size, when we use all possible Hamming-weights, as is the case for example when we load each coordinate of a classical data point by performing a one-qubit rotation with an appropriate data-dependent angle. 

Lastly, since any element of $\text{SO}(n)$ can be expressed as a product of $\mathcal{O}(n^2)$ Givens rotations, the direct sum of compound matrices can be implemented efficiently as a composition of FBS gates. Thus quantum computation can be used to efficiently parametrize and apply compound matrices over the exterior algebra.

 \begin{figure}[t!]   
 \centering
    \subfloat[Feed-forward architecture.]{
        \label{fig:feedforward}
        \includegraphics[width=0.32\textwidth]{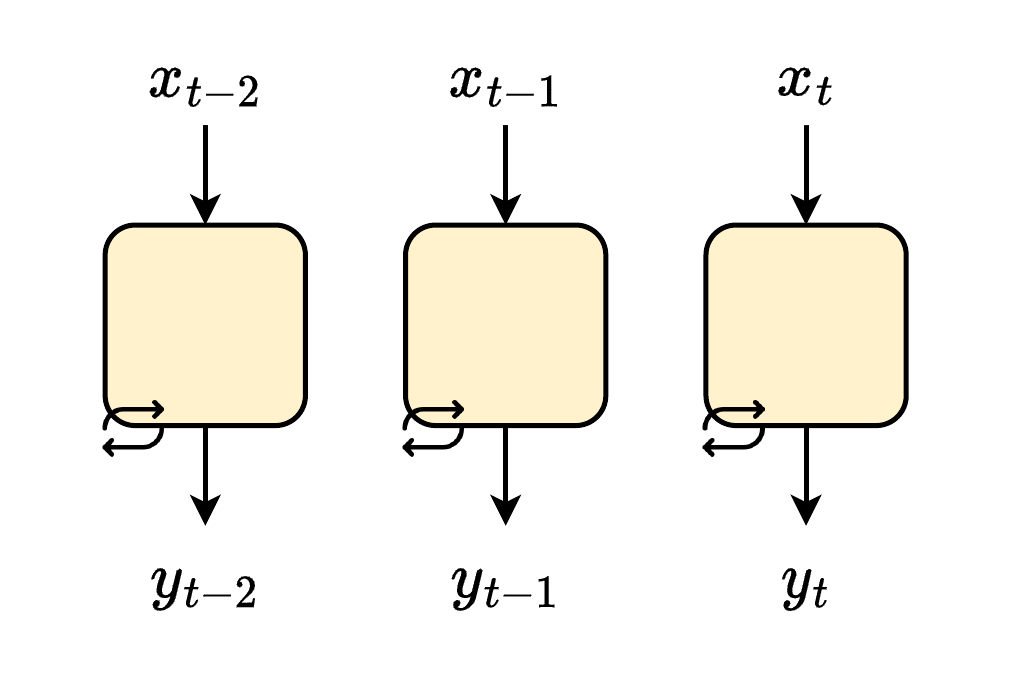}}
        \subfloat[width=0.9\textwidth][Recurrent architecture. ]{
        \label{fig:recurrent}
        \includegraphics[width=0.32\textwidth]{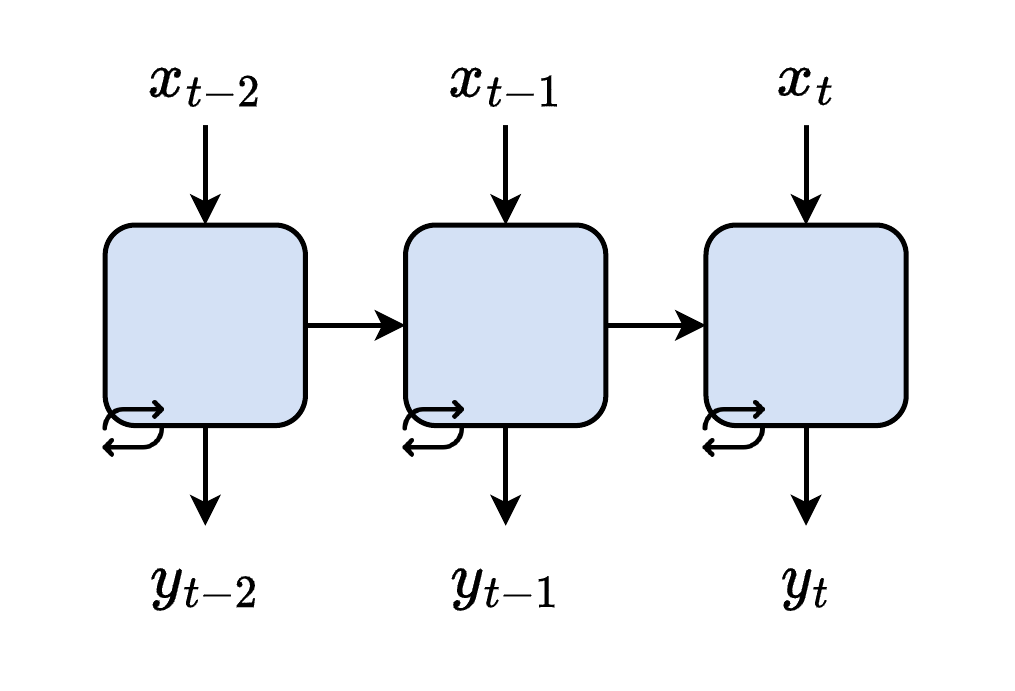}} 
        \subfloat[width=0.9\textwidth][Attention mechanism.]{
        \label{fig:attention}
        \includegraphics[width=0.32\textwidth]{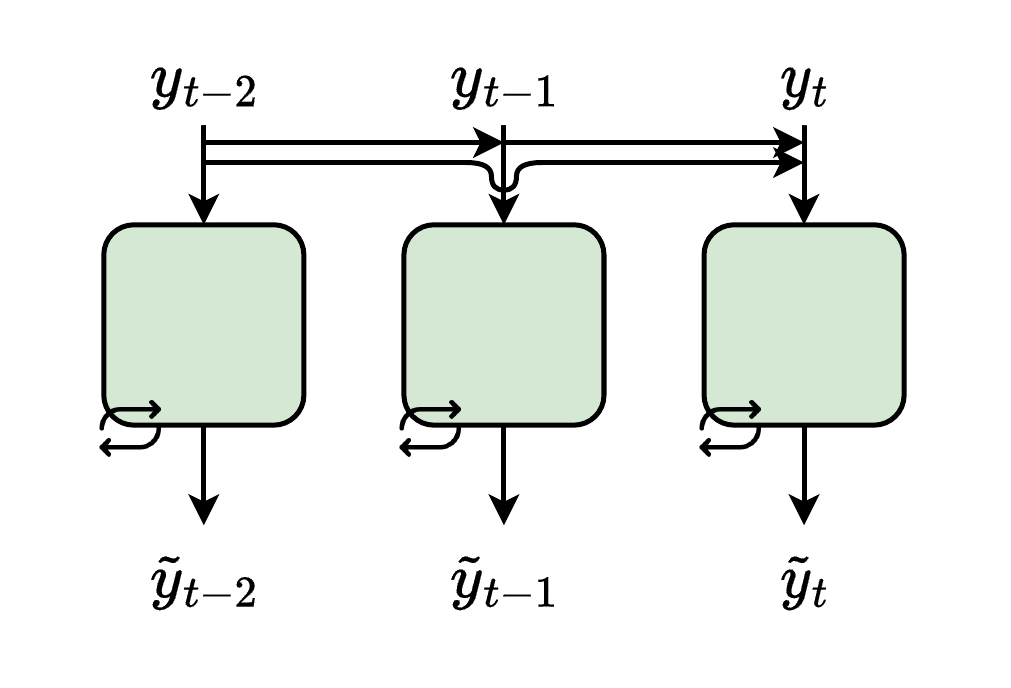}}
    \caption{Diverse quantum neural network architectures for time-series data, featuring orthogonal layers in each block as outlined in Section~\ref{subsec:qnn-archs-for-timeseries}. Here, $\boldsymbol{x}_t$ and $\boldsymbol{y}_t$ denote the time-series input and output, respectively, while $\tilde{\boldsymbol{y}}_t$ represents the output after being adjusted by the attention mechanism. }
    \label{fig:networks}
\end{figure}

\subsection{Quantum Neural Network Architectures with Orthogonal Layers}
\label{subsec:qnn-archs-for-timeseries}

Our aim is to develop quantum neural networks capable of processing sequential data. To achieve this, we will utilize classical neural network architectures that have proven to be effective in dealing with time-series data. However, we will extend these classical architectures by replacing the linear layers with our orthogonal layers. This approach offers an alternative to the use of quantum variational circuits, which are commonly used in quantum machine learning. As discussed in Section~\ref{subsec:quantum-analogues-of-classical-nns}  orthogonal layers use only the unary basis whose size is equal to the number of qubits, and hence one can easily perform tomography to obtain a classical description of the output and apply a nonlinearity. By combining the strengths of classical neural networks with the unique properties of quantum orthogonal layers, we hope to achieve improved results in processing sequential data. 

We designed quantum neural networks to process input time-series data $(\boldsymbol{x}_0, \boldsymbol{x}_1, \dots, \boldsymbol{x}_T)$ and produce the final output sequence $(\boldsymbol{y}_0, \boldsymbol{y}_1, \dots, \boldsymbol{y}_T)$. We split these architectures into two categories: feed-forward and recurrent  architectures. 

For the purpose of details, we assume that the input and output have the same dimension $n$ and that this dimension is maintained across layers. Additionally, these architectures are made up of blocks that can be repeated to create deeper architectures. Here, we will assume that the number of blocks is one.

\begin{itemize}
    \item \textbf{Feed-forward Architectures:}  A classical Feed-forward neural network consists of multiple layers of transformations, where information flows from input to output without looping back. In each layer, a linear transformation, bias shift, and a non-linear function are applied. The output is calculated as
    $$\boldsymbol{x}_t = \mathrm{f}\left(W\boldsymbol{x}_t + \beta\right),$$
    where $\mathrm{f}$ is the activation function, and $W$ and $\beta$ are the weights and biases. The number of parameters in each layer of a classical network is $\mathcal{O}(n^2)$. Our proposed quantum equivalent (Figure~\ref{fig:feedforward}) calculates each transformation as
    $$\boldsymbol{y}_t = \mathrm{f}\left(\gamma \circ U(\boldsymbol{\theta})\ket{\boldsymbol{x}_t} + \beta\right),$$
    where $\circ$ represents the element-wise product, $U(\boldsymbol{\theta})\ket{\boldsymbol{x}_t}$ is the output of a quantum orthogonal layer retrieved using tomography, $\gamma$ is a scaling factor used to rescale each feature, $\beta$ is a shift that acts as the bias, and $\mathrm{f}$ is a non-linear function such as $\mathrm{sigmoid}$, $\mathrm{tanh}$, or $\mathrm{ReLU}$. All parameters $\boldsymbol{\theta}$, $\gamma$, and $\beta$ are trainable and shared across the networks used at different time steps. The total number of trainable parameters is $\mathcal{O}(n^2)$ when using the Brick and Pyramid architectures, and $\mathcal{O}(n\log(n))$ for the Butterfly architecture. Additionally, the parameters of the quantum orthogonal layers can either be shared across layers or not, depending on the requirements of the time-series model being used. 

    \item \textbf{Recurrent Architectures:} Recurrent neural networks are designed to handle sequential data. One example of a recurrent architecture is the standard Recurrent Neural Network (RNN). In this example, we will show how to provide a quantum version of RNNs (Figure~\ref{fig:recurrent}), but the same approach can be applied to other recurrent models such as the LSTM. RNNs consist of a repeating module with a hidden state, $\boldsymbol{h}_t$, that allows information to be passed from one step of the sequence to the next. At each time-step, the network takes as input $\boldsymbol{x}_t$ and the hidden state from the previous time-step, $\boldsymbol{h}_{t-1}$. 
    The hidden state is updated as
    $$\boldsymbol{h}_t = \mathrm{f}\left(W_x \boldsymbol{x}_t + W_h \boldsymbol{h}_{t-1} + \beta \right),$$
    where $\mathrm{f}$ is an activation function, $W_x$ and $W_h$ are weight matrices, and $\beta$ is a bias vector. 
    The quantum analogue of this update is represented as
    $$\boldsymbol{h}_t = \mathrm{f}\left(\gamma_{x} \circ U(\boldsymbol{\theta}_{x})\ket{\boldsymbol{x}_t} + \gamma_h \circ U(\boldsymbol{\theta}_{h})\ket{\boldsymbol{h}_{t-1}} + \beta\right),$$
    where we now have two orthogonal layers with parameters $\boldsymbol{\theta}_{x}$ and $\boldsymbol{\theta}_{h}$, one for the input state $\boldsymbol{x}_t$ and another for the previous hidden state $\boldsymbol{h}_{t-1}$. We also have two scaling factors, $\gamma_x$ and $\gamma_h$. The hidden state $\boldsymbol{h}_t$ is then used to generate the output at each time-step, $\boldsymbol{y}_t$, through another layer, which can be implemented using another orthogonal layer to map it to the output. Moreover, the parameters are shared across layers and we can show that the number of trainable parameters per layer in the quantum Recurrent Neural Networks grows similarly to that in the quantum Feed-forward Neural Networks, with the total number of parameters being $\mathcal{O}(n^2)$ for the Brick and Pyramid architectures and $\mathcal{O}(n\log(n))$ for the Butterfly architecture. 
\end{itemize}

We also define an attention mechanism that can be applied to the output sequence $(\boldsymbol{y}_0, \boldsymbol{y}_1, \dots, \boldsymbol{y}_T)$ to create a transformer architecture. 

\begin{itemize}
    \item \textbf{Attention Mechanism:} The attention mechanism, as used in the Transformer architecture~\cite{vaswani_attention_2017}, can be applied to both Feed-forward and Recurrent Neural Networks. Here, we describe a basic quantum attention mechanism (Figure~\ref{fig:attention}), but it can be generalized. Given an output sequence $(\boldsymbol{y}_1, \boldsymbol{y}_2, \dots, \boldsymbol{y}_T)$, the goal of the attention mechanism is to compute the output $(\tilde{\boldsymbol{y}}_1, \tilde{\boldsymbol{y}}_2, \dots, \tilde{y}_T)$  as $\tilde{\boldsymbol{y}}_t=\sum_{t' \leq t} w_{t,t'} \boldsymbol{y}_{t'}$ where the weights $w_{t,t'}$ are computed by considering all previous time steps: 
    $$ w_{t,t'} \propto \exp(\boldsymbol{y}_{t'}W_y \boldsymbol{y}_{t}/\tau),$$ 
    where $W_y$ is a trainable attention matrix that combines the query and key matrices into one matrix, as described in~\cite{cherrat_quantum_2022}, and $\tau$ is a temperature parameter. 
    In the quantum case, we use
    $$ w_{t,t'} \propto \exp({\bra{\boldsymbol{y}_{t'}}U(\boldsymbol{\theta}_y)\ket{\boldsymbol{y}_t}}/\tau),$$
    where $\gamma_y$ is a scaling factor used to rescale each feature, $U(\boldsymbol{\theta}_y)$ are the parameters of the quantum orthogonal layer that computes the attention weights $w$. The dot product between $\ket{\boldsymbol{y}_{t'}}$ and $U(\boldsymbol{\theta}_y)\ket{\boldsymbol{y}_t}$ can be computed quantumly using an additional data-loader to unload $\boldsymbol{y}_{t'}$ after applying $U(\boldsymbol{\theta}_y)$ to $\boldsymbol{y}_t$. This procedure is similar to the one described in~\cite{johri_nearest_2021}.
\end{itemize}

\begin{figure}[t!]
    \centering
  \includegraphics[width=0.56\textwidth]{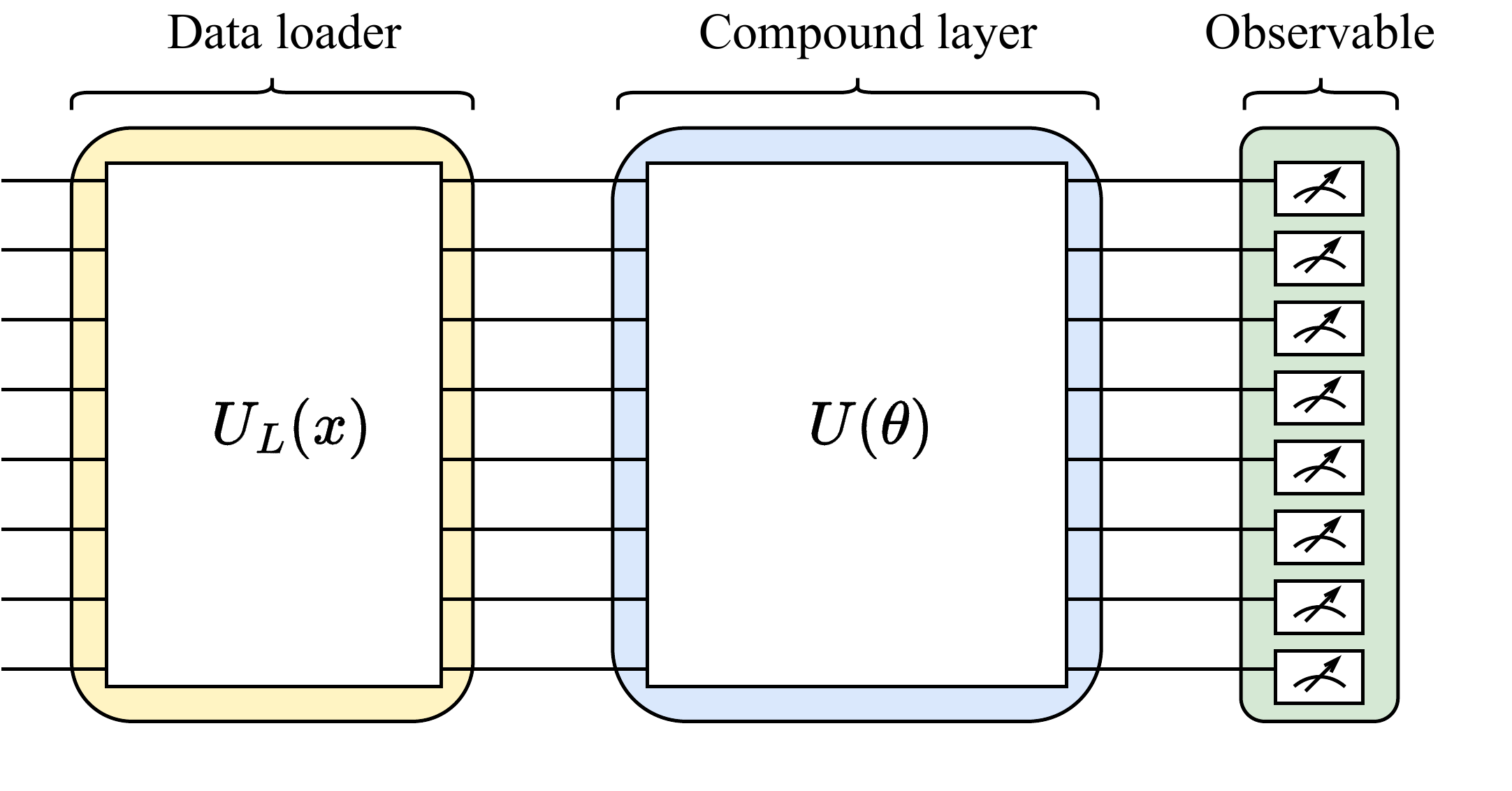}
  \caption{A quantum compound neural network. $U_{L}(\boldsymbol{x})$ refers to a general data loader unitary. $U(\boldsymbol{\theta})$ denotes a Hamming-weight preserving unitary as for example the ones shown in Figure~\ref{fig:orthogonal-layers}.}
  \label{compound}
\end{figure}

\subsection{Properties of Quantum Compound Neural Networks}
\label{subsec:quantum-compound-nns}

We define a quantum compound neural network to be the standard variational QNN, i.e. 
$$
    C(\boldsymbol{\theta}, \boldsymbol{x}) = \text{Tr}[OU(\boldsymbol{\theta})\rho(\boldsymbol{x})U^{\dagger}(\boldsymbol{\theta})],
$$
where $O$ is an observable that preserves Hamming-weight, e.g. diagonal in the computational basis), $U$ is a quantum compound layer, and $$\rho(\boldsymbol{x}) = U_L(\boldsymbol{x})[\ketbra{0}]^{\otimes n}U_L^{\dagger}(\boldsymbol{x})$$ is a quantum feature map (Figure~\ref{compound}). As mentioned in Section~\ref{sec:compound_layer}, it is sufficient to use only $\text{RBS}$ gates in the Brick architecture to implement a quantum compound layer. Thus, without loss of generality, we define $U(\boldsymbol{\theta})$ to consist only of $\text{RBS}$ gates in the Brick architecture.

Under these assumptions, the output of the QNN decomposes as
$$
    C(\boldsymbol{\theta}, \boldsymbol{x}) = \sum_{k=0}^{n}\text{Tr}[P_{k}\rho(\boldsymbol{x})]\text{Tr}[O^{(k)}A_{\boldsymbol{\theta}}^{(k)}\rho^{(k)}(\boldsymbol{x})(A_{\boldsymbol{\theta}}^{(k)})^{-1}],
$$
where $P_k$ is the projector onto the Hamming-weight-$k$ subspace, $O^{(k)} = P_{k}OP_{k}$, $\rho^{(k)} = P_{k}\rho P_{k}$, and $A_{\boldsymbol{\theta}}^{(k)}$ is the $k$-th compound matrix associated with the Givens circuit $U(\boldsymbol{\theta})$ in the manner discussed earlier. One potential application of such a subspace-preserving QNN could be the following. Suppose there is some canonical grouping of the input data $\boldsymbol{x} \in \mathcal{X}$ according to a partitioning function $f:\mathcal{X} \xrightarrow[]{} [n+1]$. Then one could potentially construct a quantum-feature map or state preparation procedure such that $P_{f(\boldsymbol{x})}\rho(\boldsymbol{x})P_{f(\boldsymbol{x})} = \rho(\boldsymbol{x})$, i.e. the quantum states encoding inputs lying in different groups are embedded into different Hamming-weight subspaces. Then it follows that
$$
      C(\boldsymbol{\theta}, \boldsymbol{x}) = \text{Tr}\left[O^{(f(x))}A_{\boldsymbol{\theta}}^{(f(x))}\rho^{(f(x))}(\boldsymbol{x})(A_{\boldsymbol{\theta}}^{(f(x))})^{\mathsf{T}}\right],
$$
and the quantum compound neural network can potentially learn different functions over the different groups. This form of learning is a special case of group-invariant machine learning, which has also recently been explored in the quantum case~\cite{larocca_group-invariant_2022}. Note that the parameters $\boldsymbol{\theta}$ are shared across the different functions which is beneficial for training. Furthermore, we show below that under Gaussian initialization, the variance of the gradient on each subspace does not vanish exponentially with the number of qubits. Thus quantum compound neural networks can be trained efficiently.

Classical neural networks over exterior algebras have been applied to manifold learning tasks, specifically for data that lies on Grassmannians~\cite{zhang_grassmannian_2018}. The $(n, k)$-Grassmannian of $V$ is a manifold containing all $k$-dimensional subspaces of $V$ and can be embedded in the space $\bigwedge^{k} V$, where $k$-wedge products of orthogonal vectors define subspaces. When considering $A \in \text{SO}(n)$, the operator $A^{(k)}$ maps between elements of the Grassmannian. While $\binom{n}{k}$ can be large, the application of $A^{(k)}$ to a Grassmannian can done by multiplying $n \times k$ and $n \times n$ matrices. However, optimizing the linear layers of the neural network while ensuring the data remains an orthogonal matrix can be computationally challenging.

Since the embedding of the $(n,k)$-Grassmannian into $\bigwedge\nolimits^{k} V$ is not surjective, and in the quantum case we can apply compound matrices to the larger space $\bigwedge\nolimits^{k} V$,  the technique used classically for Grassmannian neural networks to reduce the dimension of the matrix-vector products cannot be applied to simulate the quantum case. In other words, such compound layers are inherently quantum and perform an operation that in the general case seems to take exponential time to perform classically. Nevertheless, we will show below that such compound layers remain trainable in certain settings.

In Section~\ref{sec:quantum-deep-hedging}, we present a specific example of the Deep Hedging problem where grouping the inputs by Hamming-weight is natural in the quantum setting and improves the accuracy of the model.

Highly expressive QNNs are known to suffer from barren plateaus in their training landscape at initialization. This occurs when the variance of the gradient decays exponentially with the number of qubits, which makes sampling to estimate the gradient asymptotically intractable. Specifically, consider a typical QNN of the form
$$C(\boldsymbol{\theta}) = \text{Tr}[OU(\boldsymbol{\theta})\rho U^{\dagger}(\boldsymbol{\theta})],$$
where $O$ is an observable and $U$ is a 2-design for the Haar measure $\mu$ on  $\text{SU}(2^n)$. Using known formulas for integration over the Haar measure on compact groups, it was shown~\cite{mcclean_barren_2018} that  $\text{Var}_{\mu}[\partial_{\theta_i}C(\boldsymbol{\theta})] = \mathcal{O}(1/2^{2n})$.

More recent lines of work have shown that the symmetries of the parameterized quantum circuit also need to be considered and play an important role in understanding, for example,  convergence~\cite{you_convergence_2022, larocca_theory_2021} and trainability~\cite{larocca_diagnosing_2022} of QNNs. For trainability, it was shown that if the input state $\rho$ lies in the invariant subspace $\mathcal{H}_k$, then the variance of the gradient is in the worst case $O(1/d_{k}^2)$, where $d_k = \dim \mathcal{H}_k$. Note there will also be a dependence on the initial state used. While this result was shown for $\text{SU}$, the asymptotics of Haar moments are similar for other classical compact Lie groups~\cite{collins_integration_2006}, such as $\text{SO}$. In the case of the compound layer, the subspaces $\mathcal{H}_k$ are spanned by computational basis states with Hamming-weight $k$. Thus if $\rho \in \mathcal{H}_{k}$, where $k$ is such that $\binom{n}{k} = \mathcal{O}(\text{poly}(n))$, then the variance of the gradient does not exponentially decay with growing system size for all initial states. This is at least the case when $k$ is independent of $n$. 

It was further conjectured in~\cite{larocca_diagnosing_2022} \footnote{We note that after the writing of this manuscript, this conjecture was independently proven for certain observables~\cite{fontana2023adjoint,ragone2023unified, monbroussou2023trainability}.} that the variance actually scales with the dimension of the (dynamical) Lie algebra restricted to the invariant subspace, which can be polynomial in $n$ even when $\dim \mathcal{H}_{k}$ is exponential.  In the case of the compound layer, the dimension of the Lie algebra of the compound matrix group for $\text{SO}(n)$ is actually equal to the dimension of $\mathfrak{so}(n)$. Thus even though there are invariant subspaces whose dimension is exponential in $n$, the dimension of  the Lie algebra grows at most quadratically in $n$. If proven true, this conjecture would imply that the variance does not decay exponentially on any subspace, e.g. for $k=n/2$.

It is possible to go beyond unproven conjectures by making some well justified assumptions about the initialization and measurement phases of the quantum compound neural network, namely that (1) The parameters are randomly initialized from centered Gaussian distributions with variance inversely proportional to the number of gates in the circuit (2) The final measurement is made in the computational basis. Specifically, we make a vector valued measurement by measuring $Z_i = I^{\otimes i-1} \otimes Z \otimes I^{\otimes n  - i}$ for each qubit $i \in [1,n]$. Since the operators $Z_i$ all commute, the order of measurement is arbitrary and the measurement is well-defined. Such a measurement allows for loss functions that are arbitrary functions of the measured bit-string. It suffices therefore to demonstrate that the gradient of the output after measuring each $Z_i$ does not decay exponentially with the number of qubits. The overall gradient may still decay depending on the loss function used, but this would be a property of the loss itself and not of the quantum neural network. (3) The initial state is chosen to be either the uniform superposition, or the uniform superposition over all computational basis states of Hamming-weight $k$, where $1 < k < n$.

We now give a rigorous proof that the overall gradient does not vanish in this setting. We note that this is the setting 
in our numerical experiments~(Section~\ref{sec:applications}). We use the following theorem from~\cite{zhang_escaping_2022}, which we paraphrase below in the necessary form.

\begin{theorem}[{Paraphrased from \cite[Theorem 4.2]{zhang_escaping_2022}}]
\label{thm:gauss-vanishing}
Consider any $n$-qubit variational form with output function given by $C(\boldsymbol{\theta}) = \text{Tr}(O\Pi_{j=L}^{1}U_j(\theta_j)\rho_{\text{in}}\Pi_{j=1}^{L}U_j(\theta_j)^{\dagger})$, where $O$ is some $n$-qubit observable, and each $U_j$ is expressible as the product of a constant number of parametrized 2-qubit gates (potentially with shared parameters). Then for any parameter $\theta_j$ it holds that
\begin{align*}
\mathbb{E}_{\boldsymbol{\theta}} \left(\frac{\partial C}{\partial \theta_j}\right)^2 \ge \frac{1}{2}\left(\frac{\partial C}{\partial \theta_j}\right)^{2} \Bigg\rvert_{\boldsymbol{\theta}=0},
\end{align*}
when each $\theta_j$ is initialized from a normal distribution $\mathcal{N}(0,\gamma^2)$ with $\gamma^2 = \mathcal{O}(1/L)$.
\end{theorem}

From Theorem~\ref{thm:gauss-vanishing}, we have the following theorem.

\begin{theorem}\label{thm:gaussian_init_compound}
    Consider a variational quantum algorithm using a quantum compound layer and a final observable $Z_m$ (where $X_m,Y_m,Z_m$ correspond to the application of the corresponding Pauli gate on qubit $m$), and let the corresponding output (as a function of the parameters $\boldsymbol{\theta}$) be $C(\boldsymbol{\theta})$. We have
    $$\mathbb{E}_{\boldsymbol{\theta}} \big[\lVert \nabla C \rVert^2\big] = \Omega\left(\frac{1}{\mathrm{poly}(n)}\right),$$ 
    when the parameters are initialized from a normal distribution $\mathcal{N}(0,\gamma^2)$ with $\gamma^2 = \mathcal{O}(1/n^2)$, and the input state is chosen to be $\rho_0 =  |\psi\rangle \langle \psi |$ where $|\psi\rangle$ is an $n$-qubit state representing either the uniform superposition over computational basis states, or the equal superposition of computational basis states with Hamming-weight $k$, for any $1 \le k < n$.
\end{theorem}

\begin{proof}
    We observe that the number of parameterized gates in a quantum compound neural network is $\mathcal{O}(n^2)$. Using Theorem~\ref{thm:gauss-vanishing}, it suffices for our result to show, for some parameter $\theta_l$, that $$\frac{1}{2}\left(\frac{\partial C}{\partial \theta_l}\right)^{2} \Big\rvert_{\mathbf{\theta}=0} = \Omega\left(\frac{1}{\mathrm{poly}(n)}\right).$$
    
    All the parameterized gates in the quantum compound neural network are RBS gates. Let the gate corresponding to $\theta_l$ act on the qubits $i,j$. The corresponding unitary is
    $$ \text{RBS}_{ij}(\theta_l)  = \begin{pmatrix} 1 & 0 & 0 & 0 \\
     0 & \cos(\theta_l) & \sin(\theta_l) & 0 \\ 
     0 & - \sin(\theta_l) & \cos(\theta_l) & 0 \\
      0 & 0 & 0 & 1 \end{pmatrix}.
    $$
    It can be verified by computation that $\text{RBS}_{ij}(\theta_l) = \exp(-i \theta_l H^{\text{RBS}}_{ij})$ where,  $H^{\text{RBS}}_{ij} = \frac{Y_i \otimes X_j - X_i \otimes Y_j}{2}$. Finally define $U_{-}(\theta_{1\colon l-1}),U_{+}(\theta_{l+1\colon L})$ denote the sections of the parameterized circuit before and after the $l^{\text{th}}$ parameterized gate. By explicit differentiation, we have
    $$
        \left(\frac{\partial C}{\partial \theta_l}\right)\Bigg\rvert_{\mathbf{\theta}=0} = -i  \langle \psi | U_{-}^{\dagger} [H^{\text{RBS}}_{ij}, U_{+}^\dagger Z_m U_{+}] U_{-} |\psi\rangle \Bigg\rvert_{\mathbf{\theta}=0} = -i \langle \psi | [H^{\text{RBS}}_{ij}, Z_m]  |\psi\rangle.
    $$
    
    If $m \ne i,j$, it is clear that $[H^{\text{RBS}}_{ij}, Z_m] = 0$. We therefore consider the case when $m \in \{i,j\}$. In the rest of the proof, we assume without loss of generality that $i=0,j=1,m=0$. We have, $$[H^{\text{RBS}}_{01}, Z_0] = i(Y_0 \otimes Y_1 + X_0 \otimes X_1) = i\begin{pmatrix} 0 & 0 & 0 & 0 \\
         0 & 0 & 2 & 0 \\ 
         0 & 2 & 0 & 0 \\
          0 & 0 & 0 & 0 \end{pmatrix}.$$
    Consider any two computational basis states $|a\rangle, |b\rangle$. Clearly $-i \langle \psi | [H^{\text{RBS}}_{ij}, Z_m]  |\psi\rangle = 2$ if $a=01x, b=10x$ or vice-versa for some $n-2$ bit string $x$, and 0 otherwise. We now determine 
    $\frac{1}{2}\left(\partial C/\partial \theta_l\right)^{2}\rvert_{\mathbf{\theta}=0}$  for different possible initial states $|\psi_0\rangle$. Suppose the initial state is the uniform superposition over all computational basis states $|\psi_0\rangle = \frac{1}{2^{n/2}} \sum_{b \in \{0,1\}^n} \ket{b}$. In this case, 
    $$
        \partial_l C_k
         = \bra{\psi_0} \left( \sigma_X^{(i)} \otimes \sigma_X^{(j)} + \sigma_Y^{(i)} \otimes \sigma_Y^{(j)} \right) \ket{\psi_0} = \frac{1}{4}.
    $$
    Now let the initial state be the $\psi_k$ which is the uniform superposition over all strings of Hamming-weight $k$. For $2 \le k \le n$, we have
    $$ 
      \partial_l C_k
        = \bra{\psi_k} \left( \sigma_X^{(i)} \otimes \sigma_X^{(j)} + \sigma_Y^{(i)} \otimes \sigma_Y^{(j)} \right) \ket{\psi_k} \\
        = 2 \left( \frac{2\binom{n-2}{k-1}}{\binom{n}{k}}\right)^2 = \Omega(\frac{1}{n^6}).
    $$
    By an analogous argument for the Hamming-weight 1 subspace we find that $\partial_l C_k = \Omega(\frac{1}{n^2})$. As we have shown before, the gradients do vanish for the Hamming-weight 0 subspace.
\end{proof}

\section{Quantum Deep Hedging}
\label{sec:quantum-deep-hedging}

In this section, we present a quantum framework for Deep Hedging, referred to as \emph{Quantum Deep Hedging}, where we aim to leverage the power of quantum computing to enhance the deep reinforcement learning methods introduced in~\cite{buehler_deep_2019} for solving the hedging problem. We will incorporate the use of quantum neural networks, as defined in the previous sections, and provide quantum reinforcement learning solutions to this problem. Quantum reinforcement learning involves utilizing quantum computing to improve reinforcement learning algorithms. A comprehensive survey by Meyer et al.~\cite{meyer_survey_2022} outlines various approaches for incorporating quantum subroutines in these algorithms:

\begin{itemize}
    \item \textbf{In classical environments:} A common approach involves using quantum neural networks, such as variational quantum circuits, to represent the value function~\cite{lockwood_playing_2021, chen_variational_2020, lockwood_reinforcement_2020, kwak_introduction_2021} or the policy~\cite{jerbi_parametrized_2021, hsiao_unentangled_2022} in classical environments. These quantum neural networks can replace their classical counterparts in various reinforcement learning training methods, including value-based, policy-based, and actor-critic. Experiments have shown that they sometimes produce better policies or value estimates when applied to small environments, but can face trainability problems with larger environments due to the barren plateaus that occur in these cicuits (as discussed in Section~\ref{sec:quantum-NNs}).
    
    \item \textbf{In quantum environments:} Another approach considers the case of quantum access to the environment and aims to use this access to achieve a significant speed-up by developing a full quantum approach~\cite{cherrat_quantum_2022-1, wang_quantum_2021}. This access can be achieved by oracularizing the environment's components, such as the transition probabilities and reward function. Other methods, based on the gradient estimation algorithm from~\cite{gilyn_optimizing_2019} and  developed in~\cite{jerbi_quantum_2022, Cornelissen2018QuantumGE}, use quantum environments to directly compute the policy gradient as an output of a quantum procedure.
\end{itemize} 

When trying to apply the standard quantum reinforcement learning techniques described above to Deep Hedging, one faces some problems that need to be resolved. One issue is that most quantum neural network models for policies have only been applied to discrete action spaces, while Deep Hedging has a continuous action space with constraints. Additionally, current algorithms for training quantum policies or value functions rely on solving the discounted Bellman equation, which is not suitable for hedging as the goal is defined by a utility function and the value function no longer follows the Bellman equation. Furthermore, building a quantum-accessible environment and approximating the policy gradient with quantum methods also poses a challenge, as these methods require finite action spaces and use amplitude estimation to approximate the value function and its gradient with respect to the policy. 

Therefore, we aim to design a quantum reinforcement learning framework for Quantum Deep Hedging by addressing the challenges of standard quantum reinforcement learning techniques. In the subsequent subsections, we will outline the two methods we developed to overcome these challenges:

\begin{itemize}
    \item \textbf{Using orthogonal layers:} We explore the application of quantum reinforcement learning methods to classical Deep Hedging environments using quantum orthogonal neural network architectures. This is in contrast to prior work that used variational circuits to represent parametrized quantum policies and value functions. To compare these quantum neural network architectures, we implemented policy-search Deep Hedging to train quantum policies and solve the classical environment. By using orthogonal layers, we aim to design a straightforward and effective method for enhancing Deep Hedging with quantum computing.
    \item \textbf{Using compound layers:} We propose a quantum native approach to the Deep Hedging problem, where we formulate Quantum Deep Hedging as a fully quantum reinforcement learning problem and solve it using actor-critic methods. Following the steps in~\cite{buehler_deep_2019, buehler_deep_2019-1}, we construct a quantum environment for Deep Hedging by providing quantum representations of the environment quantities and the trading goal. We then design specific quantum neural networks using compound layers and quantum reinforcement learning algorithms to solve the problem. Our approach is inspired by distributional reinforcement learning and leverages the properties of the quantum environment to provide a model-based quantum-enhanced solution to Deep Hedging.
\end{itemize}

\begin{algorithm}[t!]
	\caption{Policy-Search Deep Hedging with Orthogonal Neural Networks}
	\label{algorithm:policy-search-deep-hedging}
	\begin{algorithmic}
		\STATE {\bfseries input} Policy QNN $\pi$.
		\STATE {\bfseries hyperparameters} Number of episodes per training step $N$.
		\STATE Initialize policy QNN with parameters $\boldsymbol{\phi}$.
		\WHILE{True}
		\FOR{episode $i=1$ {\bfseries to} $N$}
		\FOR{time-step $t=0$ {\bfseries to} $T$}
		\STATE Compute action $a^i_t: =\pi^{\boldsymbol{\phi}_t}_t(s^i_t)$. 
		\STATE Take action $a^i_t$ and receive total reward $r^i_t: = r_t(s_t^i,a_t^i)$.
		\ENDFOR
		\STATE Compute total cumulative return
		$\widetilde{R}_0^i := \sum_{t=0}^T r^i_t$.
		\ENDFOR
		\STATE Update policy parameters $\boldsymbol{\phi}$ with gradient descent to minimize
		$$\widetilde{\mathcal{L}}(\boldsymbol{\phi}) := \frac{1}{\lambda} \log\frac{1}{N}\sum_{i=1}^{N} \exp\big(-\lambda \widetilde{R}_0^i \big).$$
		\ENDWHILE
		\STATE {\bfseries output} Policy parameters $\boldsymbol{\phi}$.
	\end{algorithmic}
\end{algorithm}

\subsection{Quantum Deep Hedging in Classical Environments}
\label{subsec:quantum-deep-hedging-classical-environment}

\subsubsection{Classical Environment for Deep Hedging}

We base our classical environment on the work in~\cite{buehler_deep_2019,buehler_deep_2019-1}. To model the market state $s_t$, we assume that it can be represented by a sequence of market observations $\{M_{t}\}_{t=0}^{T}$ and provide a formal definition of the MDP for the environment described in Section~\ref{subsec:deep-hedging-environment}. At each time-step $t$, the available market information $M_{t}$ is described by $n$ numerical quantities represented by a vector: $$M_t\in\mathbb{R}^n.$$ 

This vector includes market information such as stock prices and other relevant financial data. In this setting, the market state $s_t$ can be identified with the sequence of past and actual market observations $\{M_{t'}\}_{t'=0}^{t}$ up to time-step $t$:
$$s_t = (M_0,M_1,\dots, M_t)\in\mathbb{R}^{n \times (t+1)}.$$ 
Furthermore, there are $m$ available hedging instruments, such as stocks, options, or futures, that can be traded with high liquidity in the market. The classical Deep Hedging environment can be formally defined as a finite-horizon MDP as follows:

\begin{definition}[Classical MDP for Deep Hedging]
\label{def:classical-mdp}

The classical MDP for Deep Hedging is a finite-horizon Markov Decision Process, defined by a tuple $(\mathcal{S},\mathcal{A},p,r,T)$. Here, $\mathcal{S}$ is the market state space, which can be decomposed into subsets $\mathcal{S}_t\subset \mathbb{R}^{n\times(t+1)}$ at each time-step $t$, $\mathcal{A}$ is the trading action space, which can also be decomposed into subsets $\mathcal{A}_t\subset[0,1]^m$, $p$ is the transition model that can be represented by $p_t:\mathcal{S}_t\rightarrow\Delta(\mathcal{S}_{t+1})$, $r$ is the reward function, which can be represented by $r_t:\mathcal{S}_t\times\mathcal{A}_t\rightarrow\mathbb{R}$, and $T\in\mathbb{N}^*$ is the time maturity of all hedging instruments.
\end{definition}

This formulation allows us to model the problem of Deep Hedging formally, where the objective is to choose a sequence of trading actions $\{a_t^\pi\}_{t=0}^T$ that optimizes the risk-adjusted expected returns over the given time horizon $T$, given a sequence of market observations $\{M_t\}_{t=0}^T$. At every time-step, the policy $\pi$ maps the current market state $s_t$ to an action $a^\pi_t$, performing a sequence-to-sequence mapping from $\{M_t\}_{t=0}^T$ to $\{a_t^\pi\}_{t=0}^T$.

Building upon this classical MDP framework for Deep Hedging, we now introduce quantum reinforcement learning methods specifically designed for classical environments, leveraging orthogonal layers to enhance their performance and efficiency.

\subsubsection{Quantum Reinforcement Learning methods for Classical Environments}
\label{subsec:algorithms-classical-e}
Our first approach to Quantum Deep Hedging utilizes quantum orthogonal neural network architectures to parametrize the policy $\pi$. While in this part of our work we focus on parametrizing the policy, this approach could be extended to the value function as well. 

The policy QNN $\pi(.;\boldsymbol{\phi})$ is a sequence-to-sequence model, which can be parametrized with $\boldsymbol{\phi}:=\{\boldsymbol{\phi}_t\}_{t=0}^T$, one per time-step, that can be shared or not depending on the setting, such that $\pi^{\boldsymbol{\phi}_t}_t$ represents the neural network used at time-step $t$ with parameters $\boldsymbol{\phi}_t$. The input time-series data $(M_0, M_1, \dots, M_T)$ is preprocessed classically and transformed into a sequence of embeddings $(\boldsymbol{x}_0, \boldsymbol{x}_1, \dots, \boldsymbol{x}_T)$ in a high-dimensional feature space of dimension $d$. We use the quantum neural networks described in Section~\ref{subsec:qnn-archs-for-timeseries}, which extract features from each $\boldsymbol{x}_t \in \mathbb{R}^d$ and may pass on information across time to produce the final output sequence $(\boldsymbol{y}_0, \boldsymbol{y}_1, \dots, \boldsymbol{y}_T)$. The number of hidden layers used in each neural network architecture is a hyper-parameter governed by factors like the complexity of the learning problem and the availability of resources.  The output $\boldsymbol{y}_t\in\mathbb{R}^d$ can be further processed classically to obtain the desired result, which is the action $\pi_t^{\boldsymbol{\phi}_t}(s_t)\in\mathbb{R}^m$ in this case.

We can train these parametrized quantum policies using the policy-search Deep Hedging algorithm, introduced in~\cite{buehler_deep_2019, buehler_deep_2019-1}, by updating the set of parameters $\boldsymbol{\phi}$ using gradient descent to minimize the policy loss function $\mathcal{L}(\boldsymbol{\phi})$ defined as
$$ \mathcal{L}(\boldsymbol{\phi}):=  \frac{1}{\lambda} \log \mathbb{E}_{s_0} \bigg[ \exp\Big(-\lambda \sum_{t=0}^T r_t\big(s_t, \pi_t^{\boldsymbol{\phi}_t}(\boldsymbol{s}_t)\big)\Big)\bigg|s_0\bigg].$$ 

The training procedure, outlined in Algorithm~\ref{algorithm:policy-search-deep-hedging}, proceeds as follows. At every iteration, it generates $N$ trajectories $\{s_t^i\}_{t=0}^T$ before using the policy QNN to compute the sequence of actions. Using this sequence of actions, we can compute the cumulative return for each episode and then estimate the utility over these episodes to provide an estimate $\widetilde{\mathcal{L}}(\boldsymbol{\phi})$ of the policy loss function defined earlier, and then update the parameters $\boldsymbol{\phi}$.

The same principles of using quantum neural networks that use orthogonal layers within classical architectures can be extended to other deep reinforcement learning algorithms for Deep Hedging, such as actor-critic and value-based methods, as well as future approaches in Deep Hedging.

\subsection{Quantum Deep Hedging in Quantum Environments}
\label{subsec:quantum-deep-hedging-part-2}

Here, we extend the classical Deep Hedging problem into the quantum case, starting by defining the quantum environment and the trading goal for Quantum Deep Hedging. Then, in order to develop quantum algorithms to solve this newly defined trading goal in the quantum environment, we connect to distributional reinforcement learning by showing that the value function in our environment can be expressed using a categorical distribution. 

We will propose a model-based distributional approach to approximate this value function by constructing appropriate quantum unitaries and observables that approximate the value function and its distribution. 
Furthermore, we introduce two quantum algorithms: Quantum Deep Hedging with expected actor-critic and Quantum Deep Hedging with distributional actor-critic, which are summarized in Algorithms~\ref{algorithm:expected-deep-hedging}~and~\ref{algorithm:distributional-deep-hedging}, respectively. These algorithms use quantum neural network architectures with compound layers to approximate the policy and value functions.

\subsubsection{Quantum Environment for Deep Hedging}
\label{subsec:quantum-environment}
We present here a method for converting the classical Deep Hedging problem into a quantum-native setup, which we refer to as the quantum environment for Deep Hedging. In order for this approach to be efficient, we make the state space finite, allowing the market states to be encoded into quantum states of the form $\ket{s_t}$. Moreover, at each time-step, the environment utilizes an oracle $U_t^p$ to map the transition probabilities $p_t$ to the amplitudes of a quantum state, which is a superposition of next states $\ket{s_{t+1}}$. 

 Note that the quantum environment introduced is still solving the classical Deep Hedging task and differs from previous work in quantum reinforcement learning. Previous work, such as~\cite{Cornelissen2018QuantumGE, jerbi_quantum_2022, cherrat_quantum_2022-1, wang_quantum_2021, jiang_quantum_2022}, has employed quantum environments with finite action spaces to create model-based approaches. A key distinction in our approach is that our transition oracles do not necessitate encoding of the action $a_t$, as the Deep Hedging model assumes that actions have no impact on the transition probabilities, i.e., trading actions do not affect the market\cite{murray_deep_2022}. This is a valuable feature of our approach, as it removes the need for a quantum encoding of actions. Additionally, our approach does not require quantum access to the reward function through oracles. While previous work encodes the reward function into parts of a quantum state, for example in the amplitudes~\cite{cherrat_quantum_2022-1}, in quantum registers~\cite{wang_quantum_2021, jerbi_quantum_2022}, or in the phases of the quantum states~\cite{Cornelissen2018QuantumGE}, it is unclear how to achieve this efficiently for the reward function associated with Deep Hedging and which is defined on a continuous action space. Consequently, we present an alternative formulation of quantum environments that incorporates the unique structure of the Deep Hedging MDP, as detailed in Definition~\ref{def:classical-mdp}. 

To formally specify our quantum environment, we assume that the market information $M_t$ at each time-step $t$ can be represented as a $n$-bit binary string, which we encode in an $n$-qubit computational basis state $\ket{M_t}$. Note that $\braket{M'_t}{M_t}=0$ if $M'_t\neq M_t$. The quantum encoding $\ket{s_t}$ of the market state $s_t$ can then be expressed as follows:
$$\ket{s_t}= \ket{M_0}\otimes\ket{M_1}\otimes \cdots \otimes \ket{M_t} \in \mathcal{H}^{\otimes n\times (t+1)}.$$
With the formalism described above, the oracle $U_t^p$ encoding transitions probabilities as described by a classical transition function $p_t$ can be written as follows:
$$ U_t^p: \ket{s_t}\otimes\ket{0}^{\otimes n} \xrightarrow[]{}\ket{(\boldsymbol{s}_{t+1}|s_t)}:= \sum_{s_{t+1}}\sqrt{p_t(s_{t+1}|s_t)}\ket{s_{t+1}}. $$

We will now redefine the trading goal in the context of the quantum environment for Deep Hedging. While our focus is on the exponential utility $\EE_\lambda$ as defined in Section~\ref{subsec:deep-hedging-trading-goal}, our approach can be applied to any risk measure that can be expressed as the expectation of a deterministic function over future returns. To evaluate the value function for a given policy $\pi$, we need to compute the expectation of the exponentiated rewards over future returns. Specifically, for a state $s_t$, the random variable $\exp(-\lambda R_t^\pi(\boldsymbol{s}_T))|s_t$, representing the exponentiated return, is a discrete distribution with values $\{R_t^\pi(s_T)\:|\:s_T\in\mathcal{S}_T\}$ and probabilities $p_t(s_T|s_t)$. We can express its expectation as the measurement of a quantum observable $O_t^{\pi}$ in the quantum state $\ket{(\boldsymbol{s}_T|s_t)}$, where
$$ O^{\pi}_t := \sum_{s_T} \exp(-\lambda R_t^\pi(s_T))\ket{s_T}\bra{s_T}, $$
and $\ket{(\boldsymbol{s}_T|s_t)}$ is defined similarly to $\ket{(\boldsymbol{s}_{t+1}|s_t)}$ as
$$\ket{(\boldsymbol{s}_T|s_t)}:=\sum_{s_T}\sqrt{p_t(s_T|s_t)}\ket{s_T} \in \mathcal{H}^{\otimes n\times (T+1)}.$$ This quantum state encodes the probabilities $p_t(s_T|s_t):=\mathbb{P}[s_T|s_t]$ in a superposition of all possible trajectories $(s_{t+1}, \dots, s_T)$ of length $T-t$ and can be prepared by sequentially applying oracles $U^p_{t}, U^p_{t+1}, \dots, U^p_{T-1}$ to $\ket{s_t}$ and $n \times (T-t)$ ancilla qubits.
Denoting $\rho_t(\boldsymbol{s}_T|s_t):=\ket{(\boldsymbol{s}_T|s_t)}\bra{(\boldsymbol{s}_T|s_t)}$, we have
$$ \text{{\normalfont Tr}}[O^{\pi}_t\rho_t(\boldsymbol{s}_T|s_t)]= \sum_{s_T} p_t(s_T|s_t) \times \exp(-\lambda R_t^\pi(s_T)) = \mathbb{E}[\exp(-\lambda R_t^\pi(\boldsymbol{s}_T)|s_t]. $$
Using this observable, we can redefine the trading goal in the quantum environment for Deep Hedging as finding an optimal policy $\pi^*$ such that
\[
    \forall t, \: \forall s_t \in \mathcal{S}_t, \quad v^*_t(s_t) = - \frac{1}{\lambda} \log \left\{ \inf_\pi  \text{{\normalfont Tr}}\left[O^{\pi}_t\rho_t(\boldsymbol{s}_T|s_t)\right] \right\}.
\]

Our next goal is to develop quantum algorithms to solve this newly defined trading goal in the quantum environment. In other words, our objective is to find the optimal policy $\pi^*$ that minimizes the logarithm of the expectation of the quantum observable $O^{\pi}_t$, which represents the exponentiated rewards over future returns.  
To design our quantum algorithm, we now connect to distributional reinforcement learning by showing that the value function in our environment can be expressed using a categorical distribution. 

\subsubsection{Connection with Distributional Reinforcement Learning}

The connection between our proposed Quantum Deep Hedging approach and distributional reinforcement learning lies in the definition of value functions. In distributional reinforcement learning, the focus is on learning the probability distribution of the returns, as opposed to just the expected return value. This is done using neural networks~\cite{bellemare_distributional_2017, dabney_distributional_2018} that approximate the return distribution using categorical distributions. Similarly, in our Quantum Deep Hedging approach, the quantum observable $O^{\pi}_t$ can be interpreted as a categorical distribution over all possible future returns, and our algorithms are designed to approximate this distribution and find the optimal policy $\pi^*$ that minimizes the logarithm of the expectation of this distribution. Therefore, distributional reinforcement learning provides a useful framework for our approach.

In distributional reinforcement learning, value functions are defined using distributions, and categorical distributions are a common choice in many approximation schemes.  A categorical distribution $\mathcal{P}^{\mathbf{z}}$, with a finite support $\mathbf{z}=\{z_1,z_2,\dots,z_K\}$, is defined as a mixture of Dirac measures on each element of $\mathbf{z}$ and has the form $\mathcal{P}^{\mathbf{z}} := \sum_{i} p_i \delta_{z_i}$, where $p_i\geq 0$, $\sum_i p_i = 1$, and $\delta_{z_i}$ is the Dirac measure on $z_i$~\cite{lyle_comparative_2019}. In other words, a categorical distribution is a probability distribution over a finite set of discrete outcomes. To formalize the notion of approximating distributions, we will use the Cram\'{e}r distance (or $\ell_2$ metric) defined as follows:

\begin{definition}[Cram\'{e}r distance]
    Given two distributions $\mathcal{P},\mathcal{Q}$ over subsets of $\mathbb{R}$, with cumulative distribution functions (over $\mathbb{R}$) given by $F_{\mathcal{P}},F_{\mathcal{Q}}$ respectively, the Cram\'{e}r distance between the two distributions is defined as
    $$
        C_2(\mathcal{P},\mathcal{Q}) = \left(\int_\mathbb{R} |F_{\mathcal{P}}(x) - F_{\mathcal{Q}}(x)|^2 \,dx\right)^{1/2}.
    $$
\end{definition}

\begin{figure}[t!]
    \centering
    \begin{subfigure}{\textwidth}
    \centering
    \includegraphics[width=0.6\textwidth]{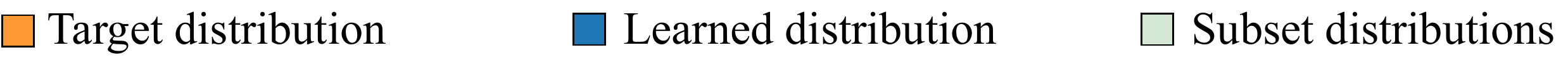}
    \end{subfigure}

    \centering
    \begin{subfigure}{0.49\textwidth}
    \centering
    \includegraphics[height=14em]{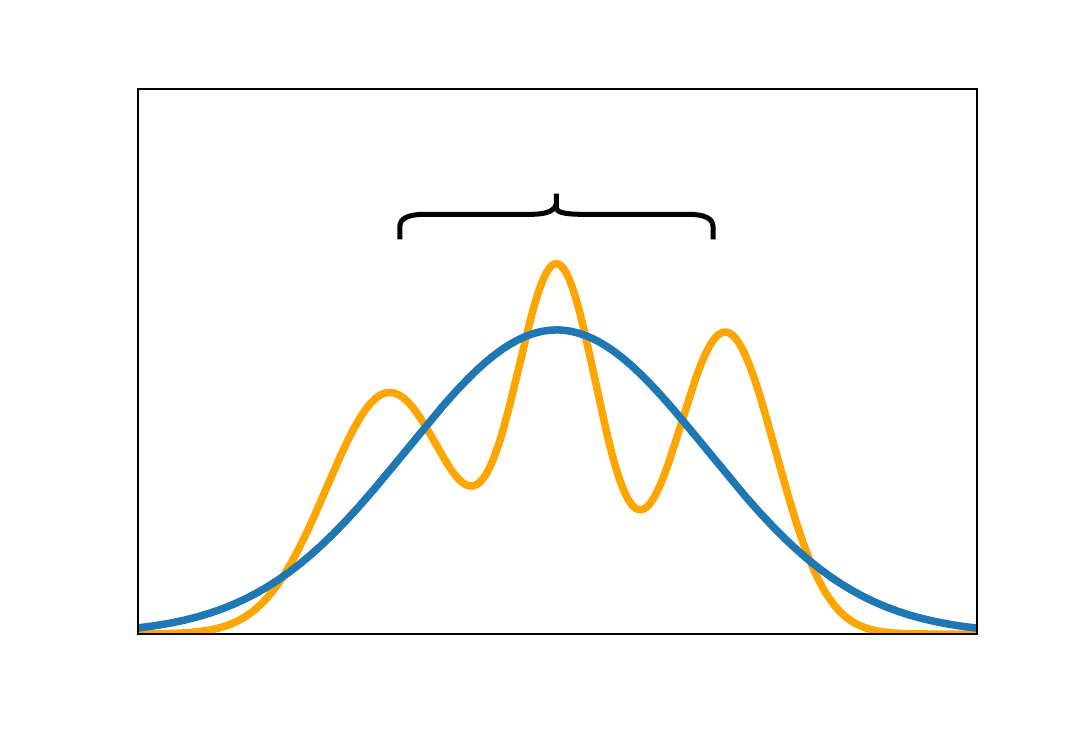} 
    \caption{\label{subfig:exp}Learning over $\mathcal{S}_T$.}
    \end{subfigure}
    \begin{subfigure}{0.49\textwidth}
    \centering
    \includegraphics[height=14em]{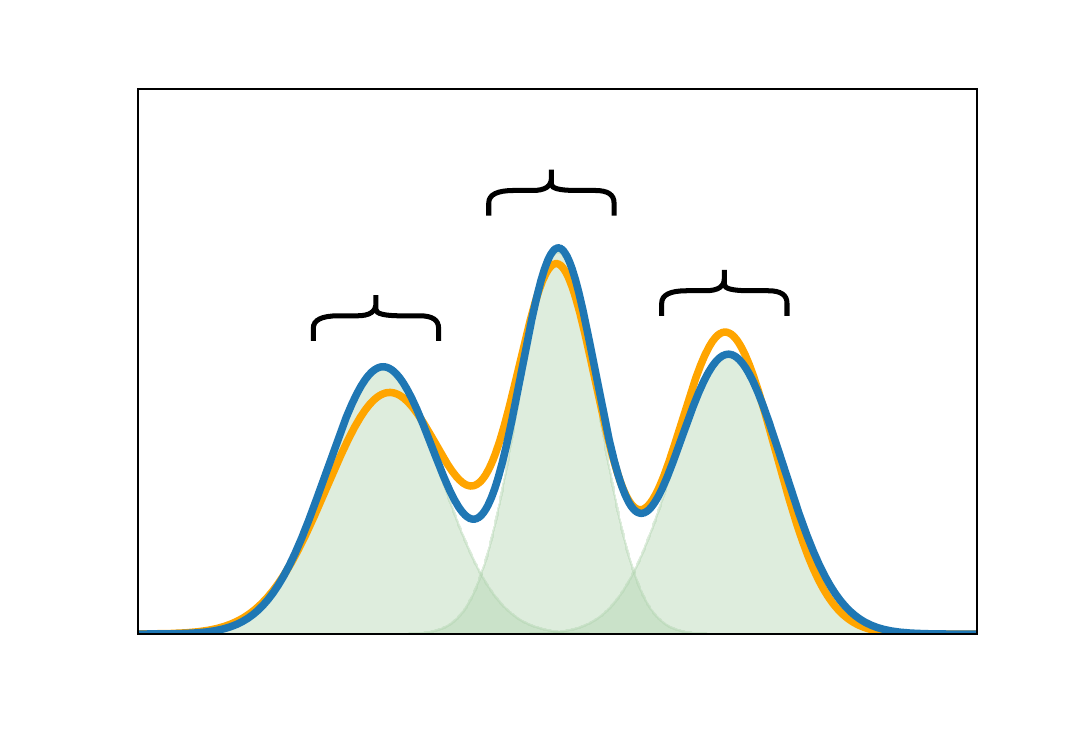}
    \caption{\label{subfig:distr} Learning over multiple subsets of $\mathcal{S}_T$.}
    \end{subfigure}
    \caption{An example illustrating the process of learning expectations for each subset. The figures depict data generated from a trimodal distribution. In (\ref{subfig:exp}), a single distribution is learned to match the expectation over the entire set of states. In (\ref{subfig:distr}), the set is divided into three distinct subsets, with each peak representing the weighted and learned distribution within its corresponding subset. This improved approach provides a closer fit to the original data and effectively incorporates information about tails, offering a more accurate representation of the underlying distribution.}
    \label{fig:bimodal}
\end{figure}

An important result from the distributional reinforcement learning literature~\cite[Proposition 1]{lyle_comparative_2019} shows the following properties of the Cram\'{e}r distance: A categorical distribution $\mathcal{P}^{\mathbf{z}}$ over some support $\mathbf{z}$ can be projected onto a categorical distribution over a different support $\mathbf{z}'$ by a mapping $\Pi_C$ (called the Cram\'{e}r projection) that preserves the expectation ($\mathbb{E}[\mathcal{P}^{\mathbf{z}}]=\mathbb{E}[\Pi_C(\mathcal{P}^{\mathbf{z}})]$) while minimizing the Cram\'{e}r distance between $\mathcal{P}^{\mathbf{z}}$ and $\Pi_C(\mathcal{P}^{\mathbf{z}})$, as long as the new support $\mathbf{z}'$ has a larger range, i.e $[\min\mathbf{z},\max\mathbf{z}]\subset[\min\mathbf{z}',\max\mathbf{z}']$.

To establish a connection between our framework and distributional reinforcement learning, we can utilize the fact that sampling from a categorical distribution can be achieved by measuring a quantum observable. Specifically, the observable $O^{\mathbf{z}}$, defined as 
$$O^{\mathbf{z}}:=\sum_{b\in\{0,1\}^m}z_b \ket{b}\bra{b},$$ 
can be applied on the quantum state $\ket{\mathbf{z}}:=\sum_{b}\sqrt{p_b}\ket{b}$ to sample from the categorical distribution $\mathcal{P}^{\mathbf{z}}$ with support $\mathbf{z}$. Here, the number of qubits $m$ required to index all the outcomes is such that $K\leq 2^m$. Thus, measuring $O^\mathbf{z}$ in $\ket{\mathbf{z}}$ serves as a quantum representation of the distribution $\mathcal{P}^{\mathbf{z}}$ over the support $\mathbf{z}$.

In the Quantum Deep Hedging framework, the categorical distribution that models the returns distribution of a policy $\pi$ at a time-step $t$ can be represented by measuring $O_{t}^{\pi}$, the quantum observable defined in Section~\ref{subsec:quantum-environment}, in the quantum state $\ket{(\boldsymbol{s}_T|s_t)}$ that encodes the probabilities of the future trajectories given the current state $s_t$. 

Constructing the quantum state $\ket{(\boldsymbol{s}_T|s_t)}$ requires quantum access to the environment, and it can be generated using the transition oracles $\{U^p_{t'}\}_{t'=t}^{T}$ and measured with the observable $O^{\pi}_{t}$. However, the observable $O^{\pi}_t$ is not efficient to implement as its natural description requires classical computation of its eigenvalues and is dependent on the policy $\pi$, which changes during the training procedure. To address this, we propose a construction to approximately represent the distribution using a much simpler observable that is diagonal in the computational basis and has a spectrum independent of the specific reward structure.  
\newpage
This construction is detailed in the following proposition that demonstrates,  using the Cram\'{e}r projection, how an observable with a fixed support can approximate the value distribution while preserving its expectation. 

\begin{proposition}
    \label{prop:cramer}
    Consider a support $\mathbf{z}$ of size $2^m$ such that, for any policy $\pi$, the following holds:
    $$ \forall \: s_T \in\mathcal{S}_T, \quad  \min_{b\in\{0,1\}^{m}}z_b \leq \exp(-\lambda R^\pi_t(s_T)) \leq \max_{b\in\{0,1\}^{m}}z_b,$$
    where $m\geq 1$. Then, there exists an observable $O^{\mathbf{z}}_t$ with eigenvalues in $\mathbf{z}$ that operates on $n\times(t+1) +m$ qubits and such that, for any deterministic policy $\pi$, there is a unitary $U^\pi_t$ satisfying
    $$ \forall s_t \in \mathcal{S}_t, \quad \text{\normalfont Tr}[O^{\mathbf{z}}_t\rho_t(\mathbf{z}|s_t) ] = \text{{\normalfont Tr}}[O^{\pi}_t\rho_t(\boldsymbol{s}_T|s_t)], $$
    where $\ket{(\mathbf{z}|s_t)}:=U_t^\pi (\ket{s_t}\otimes\ket{0}^{\otimes m})$ and $\rho_t(\mathbf{z}|s_t):=\ket{(\mathbf{z}|s_t)}\bra{(\mathbf{z}|s_t)}$.

    Additionally, let $\mathcal{P}^\pi_t$ and $\mathcal{P}^{\mathbf{z}}_t$ denote the distributions of the value function (outcome of measuring $O^{\pi}_{t}$ in $\ket{(\boldsymbol{s}_T|s_t)}$) and the corresponding categorical projection onto the fixed support $\mathbf{z}$ (outcome of measuring $O^{\mathbf{z}}_t$ in $\ket{(\mathbf{z}|s_t)}$), respectively. If the cumulative distribution function of $\mathcal{P}^\pi_t$ is $L$-Lipschitz, then the Cram\'{e}r distance between $\mathcal{P}^\pi_t$ and $\mathcal{P}^{\mathbf{z}}_t$ is such that: $$C_{2}(\mathcal{P}^\pi_t,\mathcal{P}^{\mathbf{z}}_t) \leq LZ^{3/2}/3 \cdot 2^m,$$ where $Z = (\max_{b\in\{0,1\}^{m}}z_b - \min_{b\in\{0,1\}^{m}}z_b)$.
\end{proposition}
\begin{proof}

Given a policy $\pi$ and a state $s_t\in\mathcal{S}_t$, the distribution of the exponentiated returns is a categorical distribution, denoted as $\mathcal{P}^\pi_t$, from which we can obtain a categorical distribution over a support $\mathbf{z}$ using the Cram\'{e}r projection. 
Specifically, we can project $\mathcal{P}^\pi_t$ onto the support $\mathbf{z}$ to obtain $\mathcal{P}^{\mathbf{z}}_t:=\Pi_{C}(\mathcal{P}^\pi_t)$, which returns a $z_b\in\mathbf{z}$ with probability $p(z_b|s_t)$. We can then define a unitary that maps $\ket{s_t}\ket{0}^{\otimes m}$ to $\ket{s_t}\ket{(\mathbf{z}|s_t)}$, where $\ket{(\mathbf{z}|s_t)}:=\sum_{b}\sqrt{p(z_b|s_t)}\ket{b}$ is a quantum state encoding the probabilities of $\mathcal{P}^{\mathbf{z}}_t$. The claim for a particular state $s_t$ is satisfied when we measure this unitary with the observable $$ O^{\mathbf{z}}_t := I^{\otimes n\times (t+1)} \otimes  \sum_{b\in\{0,1\}^n} z_b \ket{b}\bra{b},$$where $I$ is the identity operator acting on one qubit. Since the different quantum encodings $\ket{s_t}$ are orthogonal, a unitary $U_t^\pi$ that performs the Cram\'{e}r projection for all $s_t\in\mathcal{S}_t$ can be constructed. By known properties of the Cram\'{e}r projection, the first requirement of matching expectations is satisfied.

We now consider the Cram\'{e}r distance between the true and projected distributions. As the Cram\'{e}r projection minimizes Cram\'{e}r distance we obtain an upper bound by analyzing the distance from any projection onto the same support. Consider the projection assigns to each $z_b$ the weight of the true distribution between $z_b$ and $z_{b-1}$ (where the subtraction is performed by viewing $b$ as the binary representation of an integer). Let the true and projected cdfs be $F_{\mathcal{P}^{\pi}_t}$ and $F_{\mathcal{P}^{\mathbf{z}}_t}$ respectively. The square of the Cram\'{e}r distance between these distributions is the sum of $2^{m}$ integrals of the form: $$\int_{z_b}^{z_b + Z/2^m} (F_{\mathcal{P}^{\pi}_t}(x) - F_{\mathcal{P}^{\mathbf{z}}_t}(z_b))^2 dx \le L^2 Z^3/ 3\cdot2^{3m}.$$ Accumulating the $2^m$ terms and taking the square root, we have the necessary bound on the Cram\'{e}r distance.
\end{proof}

Proposition~\ref{prop:cramer} illustrates the use of quantum circuits to approximate the distributional value function by fixing a support $\mathbf{z}$, and measuring the observable $O_t^{\mathbf{z}}$ in the quantum states $\ket{(\mathbf{z}|s_t)}$ produced by the unitary $U_t^\pi$. We can hope to learn $U_t^\pi$ by using existing classical distributional reinforcement learning algorithms for our quantum setting. However, this approach may be challenging for two reasons. First, most of the existing approaches in distributional reinforcement learning assume that the value function conforms to the discounted Bellman equation, which is not our case since we need to take into account the risk-adjusted measure. Second, the size of the quantum support increases exponentially with the number of qubits, making training impractical.

To overcome these challenges, we propose a new approach that utilizes the environment's model and exploits the structure and properties of our Hamming-weight preserving unitaries. This approach allows us to learn polynomial-sized distributions rather than exponential ones.

\subsubsection{Distributional Value Function Approximation}
\label{sec:distributional-value}

We propose a model-based approach to learn distributional value functions in quantum environments. We introduce a new distributional reinforcement learning algorithm that differs from existing methods, in particular from~\cite{bellemare_distributional_2017} that fixes the support $\mathbf{z}$ and learns the probabilities for each element in the support and from~\cite{dabney_distributional_2018} that fixes the probabilities (or quantiles) and learns the support. Our approach in fact splits the set of future trajectories into subsets and learns the expectation of the distribution within each subset. Because our approach is model-based, we can use the model to compute the probabilities of these subsets and calculate the overall expectation as well. 

This approach can use any type of unitaries that preserve subspaces of some nature, while here we give a specific example using the Hamming-weight preserving quantum compound neural networks developed in Section~\ref{subsec:quantum-compound-nns}. As we have seen, these compound neural networks preserve the Hamming-weight subspaces and thus allow us to split the set of future trajectories according to the Hamming-weight of their quantum representation.  

In more detail, we partition the set of all complete trajectories $\mathcal{S}_T$ into $n(T-t)+1$ disjoint subsets $\mathcal{S}_T^{t,k}$, where $k=0,\dots,n(T-t)$, such that each subset contains terminal states with a Hamming-weight of $k$ in their trajectories from $t+1$ to $T$. This allows us to decompose the superposition $\ket{(\boldsymbol{s}_T|s_t)}$ of all trajectories by grouping the future trajectories by Hamming-weight. Specifically, we express $\ket{(\boldsymbol{s}_T|s_t)}$ as a sum of terms corresponding to each subset, with each term given by $\ket{(\boldsymbol{s}_T|s_t,k)}$, where $k$ is the Hamming-weight of the trajectories in the subset. We learn one expectation per subset, and by computing the probability of each subset, we can recover the overall expectation over all subsets. The probability that the future trajectory will have Hamming-weight $k$ is denoted by $p_t(s_t,k):=\mathbb{P}[|\boldsymbol{s}_T|-|s_t|=k]$, and $\ket{(\boldsymbol{s}_T|s_t,k)}$ is the superposition of such trajectories defined as: $$\ket{(\boldsymbol{s}_T|s_t,k)} := \frac{1}{\sqrt{p_t(s_t,k)}}\sum_{s_T\in\mathcal{S}_{T}^{t,k}} \sqrt{p_t(s_T|s_t)}\ket{s_T}.$$

In what follows, we will show that there exist Hamming-weight preserving unitaries that can approximate the expected value for each subset of future trajectories grouped by Hamming-weight, extending the result of Proposition~\ref{prop:cramer}.

\begin{proposition}
    \label{prop:cramer-bis}
    Consider a support $\mathbf{z}$ of size $2^{n(T-t)+2}$
    such that, for any policy $\pi$ and for any Hamming-weight $k=0,\dots,n(T-t)$, the following holds:
    $$\forall \: s_T \in\mathcal{S}^{t,k}_T, \quad  \min_{|b|=k+1}z_b \leq \exp(-\lambda R^\pi_t(s_T)) \leq \max_{|b|=k+1}z_b$$ 
    Then, there exists an observable $O^{\mathbf{z}}_t$ with eigenvalues in $\mathbf{z}$ that operates on $n\times(T+1)+2$ qubits and such that, for any deterministic policy $\pi$, there is a Hamming-weight preserving unitary $U^\pi_t$ satisfying:
    $$ \forall k, \:\: \forall s_t \in \mathcal{S}_t, \quad \text{\normalfont Tr}[O^{\mathbf{z}}_t \rho_t(\mathbf{z}|s_t,k+1)] = \text{{\normalfont Tr}}[O^{\pi}_t\rho_t(\boldsymbol{s}_T|s_t,k)] $$
    where $\ket{(\mathbf{z}|s_t,k+1)}:=U_t^\pi(\ket{(\boldsymbol{s}_T|s_t,k)}\otimes\ket{01})$ and $\rho_t(\mathbf{z}|s_t,k+1):=\ket{(\mathbf{z}|s_t,k+1)}\bra{(\mathbf{z}|s_t,k+1)}$. 
\end{proposition}

\begin{proof}
Given a policy $\pi$ and a state $s_t\in\mathcal{S}_t$, the distribution of the exponentiated returns restricted to future paths with Hamming-weight $k$ is also a categorical distribution, denoted as $(\mathcal{P}^{\pi}_t)^{(k)}$, from which we can obtain a categorical distribution $(\mathcal{P}^{\mathbf{z}}_t)^{(k)}$ over a support $\mathbf{z}^{(k)}:=\{z_b \: | \: |b|=k+1\}$ using the Cram\'{e}r projection.
We can construct a Hamming-weight preserving unitary that maps $\ket{(\boldsymbol{s}_T|s_t,k)}\otimes\ket{01}$ to $\ket{s_t}\otimes\ket{(\mathbf{z}^{(k)}|s_t)}$, where: $\ket{(\mathbf{z}^{(k)}|s_t)}$ encodes the probabilities of $(\mathcal{P}^{\mathbf{z}}_t)^{(k)}$. Measuring this state with the observable $O^{\mathbf{z}}_t$ as defined in the proof of Proposition~\ref{prop:cramer} satisfies the claim for a specific state $s_t$ and Hamming-weight $k$. Since the quantum states $\ket{(\mathbf{z}^{(k)}|s_t)}$ have different Hamming-weights, it is possible to construct a unitary operation that performs this mapping for all Hamming-weights. Similarly, as in Proosition~\ref{prop:cramer}, we can construct a unitary that performs this mapping for all states $s_t\in\mathcal{S}_t$.
\end{proof}

Regarding the above proof, there are two points to note. First, in order to calculate expectations on different subsets of future trajectories grouped by Hamming-weight, we added two ancilla qubits to satisfy the requirement of the Cram\'{e}r projection for a support of size at least $2$. This requirement is not satisfied for trajectories of Hamming-weight $0$ or $n(T-t)$ without the addition of these ancilla qubits. As a result, the Hamming-weight of the measured quantum states has been shifted by $+1$. Second, if the subsets of trajectories with Hamming-weights $0$ and $n(T-t)$ are empty, then only $n(T-t)$ qubits are needed instead of $n(T-t)+2$.

If we have Hamming-weight preserving unitaries, as described in Proposition~\ref{prop:cramer-bis}, that produce the correct expectation on every subset, then we can obtain the overall expectation of the distributional value function. By loading the superposition over all future states $\ket{(\boldsymbol{s}_T|s_t)}$ using the transition oracles $\{U_{t'}^p\}_{t'=t}^{T-1}$ and applying the unitary $U_t^\pi$ from Proposition~\ref{prop:cramer-bis}, we can calculate directly the overall expectation without having to reconstruct the expectations per subspace and compute the overall expectation classically. Hamming-weight preserving unitaries act only inside the subspace to map each $\ket{(\boldsymbol{s}_T|s_t,k)}\otimes\ket{01}$ to $\ket{(\mathbf{z}|s_t,k+1)}$. Therefore, the overall expectation matches the expectation of the value function when we apply $U_t^\pi$ to $\ket{(\boldsymbol{s}_T|s_t)}\otimes\ket{01}$, resulting in the state $\ket{(\mathbf{z}|s_t)}$ with density $\rho^{\pi}_t(\mathbf{z}|s_t))$. In other words, we have
\begin{align*}
    \text{\normalfont Tr}[O^{\mathbf{z}}_t \rho^{\pi}_t(\mathbf{z}|s_t))]  & = \sum_{k=0}^{n(T-t)}  \text{{\normalfont Tr}}[P_{k+1}\rho_t^{\pi}(\mathbf{z}|s_t)]\times  \text{{\normalfont Tr}}[O^{\pi}_t\rho_t(\mathbf{z}|s_t,k+1)] \\
    & = \sum_{k=0}^{n(T-t)}  \text{{\normalfont Tr}}[P_{k}\rho_t(\boldsymbol{s}_T|s_t)] \times \text{{\normalfont Tr}}[O^{\pi}_t\rho_t(\boldsymbol{s}_T|s_t,k)] \\
      & = \text{{\normalfont Tr}}[O^{\pi}_t\rho_t(\boldsymbol{s}_T|s_t)].
\end{align*}

We have shown that for any policy $\pi$, a Hamming-weight preserving unitary exists at each time-step $t$, which can predict the expectation of the exponentiated returns on every subset of future paths accurately, as well as the overall expectation over all future paths by projecting the output states onto an observable $O_t^{\mathbf{z}}$ independent of $\pi$. We will now use quantum compound neural networks to parametrize these unitaries and provide reinforcement learning algorithms to learn the expectation on every subset before using the overall expectation to improve the policy $\pi$.

\subsubsection{Compound Neural Networks for Deep Hedging}
\label{subsec:compound-circuit-construction}

In this section, we will discuss a general approach to constructing quantum neural networks using Hamming-weight preserving unitaries that can be used in Quantum Deep Hedging for quantum environments. We will explain how these networks can be used for both the policy and value function before providing algorithms for training them in Section~\ref{subsec:algorithms-quantum-e}. At each time-step, we design a quantum neural network that acts on $n(T+1)+2$ qubits. It takes the quantum encoding $\ket{s_t}$ as input on the first $n(t+1)$ qubits, applies the transition oracles on $n(T+1)$ qubits, and uses a compound neural network with two additional ancilla qubits (initialized to $\ket{01}$) to predict a value within some range $(o_{\min},o_{\max})$. First, we will design the observable and then present the architecture of the quantum circuit for the value function, which can also be applied to the policy by adjusting the bounds accordingly.

To construct the observable, we need to define a support that covers the valid range of values $(o_{\min}, o_{\max})$ for every subset of Hamming-weights of size at least $2$. For the value function, this range corresponds to the possible values of the exponentiated returns, which is crucial for the validity of the Cram\'{e}r projection. For the policy, the support should only cover the range of valid actions. To design the support, we draw inspiration from~\cite{bellemare_distributional_2017} and define the support for each subset of Hamming-weight $k+1$ as having size $N_k=\binom{n(T-t)+2}{k+1}$ with uniformly spaced atoms. We can express this support as
$$ O^{\mathbf{z}}_{t} = I^{\otimes n(t+1)}\otimes \sum_{k=0}^{n(T-t)}\sum_{|b|=k+1} \big(o_{\min} + (o_{\max}-o_{\min})\frac{i_b}{N_k-1}\big)\ket{b}\bra{b},$$
where $i_b$ ranges from $0$ to $N_k-1$ and indexes all the quantum states with Hamming-weight $k+1$. 

To construct the unitary, we begin with $\ket{s_t}\otimes\ket{0}^{\otimes n(T-t)}\otimes\ket{01}$ and apply the model to create a superposition over all future trajectories $\ket{(\boldsymbol{s}_T|s_t)}\otimes\ket{01}$. We can rewrite this superposition as the tensor decomposition $\ket{s_t}\otimes\ket{\boldsymbol{M}_{t+1},\dots,\boldsymbol{M}_T}\otimes\ket{01}$, since non-zero probability paths necesarily start from $s_t$. Next, we conditionally apply a compound layer to the last $n(T-t)+2$ qubits, controlled by the first $n(t+1)$ qubits. Instead of learning $|\mathcal{S}_t|=\mathcal{O}(2^{nt})$ parameters for each $s_t$, we control each of the first $n(t+1)$ qubits individually and learn a maximum number of $2n(t+1)$ parameters, two sets per control qubit. 

Since the control qubits are always in a computational basis state, they can be removed and the appropriate unitary can be chosen classically for each past state. Thus, our compound neural network $U_t(\boldsymbol{\omega}_t)$ with parameters $\boldsymbol{\omega}_t$ is written as the product of $n(t+1)$ compound layers, each with parameters $\boldsymbol{\omega}_t^{i,0}$ or $\boldsymbol{\omega}_t^{i,1}$ depending on whether the $i$-th qubit is in state $\ket{0}$ or $\ket{1}$.

After applying the unitary $U_t(\boldsymbol{\omega}_t)$ as described earlier, we obtain the quantum state with density $$\rho_t^{\boldsymbol{\omega}_t}(\mathbf{z}|s_t) = \ket{(\mathbf{z}|s_t)}\bra{(\mathbf{z}|s_t)},$$ where $\ket{(\mathbf{z}|s_t)} = U_t(\boldsymbol{\omega}_t)\ket{(\boldsymbol{s}_T|s_t)}\otimes\ket{01}$. Measuring the observable $O_t^{\mathbf{z}}$ in this state results in a value within the range $(o_{\min}, o_{\max})$. Furthermore, when using this circuit for the value function, we can apply $U_t(\boldsymbol{\omega}_t)$ to a specific Hamming-weight in order to predict a value only for that subset of trajectories.

In summary, we have described how to construct a quantum neural network for a given time-step $t$ using the environment's model. These quantum neural networks will be different for different environments and we provide a concrete example in Section~\ref{subsec:results-quantum-environment} (see Figure~\ref{ansatz-controlled}).
In the next section, we will use these networks to model the value function, adjust the same compound neural network such that it can be used for the policy and train the parameters of both networks using an actor-critic algorithm.

\subsubsection{Quantum Reinforcement Learning Methods for Quantum Environments}
\label{subsec:algorithms-quantum-e}

We present an actor-critic approach to Quantum Deep Hedging, which is specifically designed for quantum environments. To represent the policy and the value function, we use compound neural networks, which were introduced in the previous section. For each time-step $t$, we use a compound neural network with $n(t+1)$ layers, where each layer is controlled by one of the first $n(t+1)$ qubits and acts on the remaining $n(T-t)+2$ qubits. We use the Brick architecture with logarithmic depth for each layer, resulting in an overall depth of $\mathcal{O}(nt\log n(T-t))$.

The value QNN $v$ is parameterized by $\boldsymbol{\omega}:=\{\boldsymbol{\omega}_t\}_{t=0}^{T}$, where $\boldsymbol{\omega}_t$ contains all the parameters used at time-step $t$. Each layer's parameters, $\boldsymbol{\omega}_t^i$, can be further split depending on the possible values of the $i$-th control qubit. We construct the observable $O^{\mathbf{z}}_t$ as defined in Section~\ref{subsec:compound-circuit-construction} on a support $\mathbf{z}$ that covers the bounds on the exponentiated return function. We assume knowledge of $\mathbf{z}$, which can be computed classically. In this case, the value QNN maps the state $s_t$ to
$$ v^{\boldsymbol{\omega}_t}_t(s_t):=-\frac{1}{\lambda} \log \text{\normalfont Tr}[O^{\mathbf{z}}_t\rho^{\boldsymbol{\omega}_t}_t(\mathbf{z}|s_t) ].$$
Here, $\rho^{\boldsymbol{\omega}_t}_t(\mathbf{z}|s_t):=\ket{(\mathbf{z}|s_t)}\bra{(\mathbf{z}|s_t)}$, and $\ket{(\mathbf{z}|s_t)} = U_t(\boldsymbol{\omega}_t)\ket{(\boldsymbol{s}_T|s_t)}\otimes\ket{01}$ is the output of the value QNN when applied to $s_t$. Similarly, we define the policy network $\pi$ with parameters $\boldsymbol{\phi}:=\{\boldsymbol{\phi}_t\}_{t=0}^{T}$. For each time-step, we define an observable $O^{\mathbf{a}}_t$ on a different support $\mathbf{a}$ that covers the range of valid actions (typically $[0,1]$) using our approach described in Section~\ref{subsec:compound-circuit-construction} such that
$$ \pi^{\boldsymbol{\phi}_t}_t(s_t):= \text{\normalfont Tr}[O^{\mathbf{a}}_t\rho^{\boldsymbol{\phi}_t}_t(\mathbf{a}|s_t) ],$$
where $\rho^{\boldsymbol{\phi}_t}_t(\mathbf{a}|s_t):=\ket{(\mathbf{a}|s_t)}\bra{(\mathbf{a}|s_t)}$, and $\ket{(\mathbf{a}|s_t)} = U_t(\boldsymbol{\phi}_t)\ket{(\boldsymbol{s}_T|s_t)}\otimes\ket{01}$ is the output of the policy QNN when applied to $s_t$. If there is more than one hedging instrument, we can construct one policy QNN per instrument. However, this is not necessary for the value function since it evaluates the overall policy for all instruments.

To train the value network, we use two optimization objectives: the \emph{distributional} and the \emph{expected} losses. The distributional loss $\mathcal{L}_D$ takes the Hamming-weight of the current state into account when evaluating the expected reward, while the expected loss $\mathcal{L}_E$ only considers the expected reward at time $t$. Specifically, $\mathcal{L}_D(\boldsymbol{\omega})$ is defined as
$$\mathcal{L}_D(\boldsymbol{\omega}) := \mathbb{E}_{s_t,k}\Big[\big(\text{\normalfont Tr}[O^{\mathbf{z}}_t \rho^{\boldsymbol{\omega}_t}_t(\mathbf{z}|s_t,k+1))] - \exp{(-\lambda R^\pi_t(\boldsymbol{s}_T))}\big)^2\big|s_t,k],$$
and $\mathcal{L}_E(\boldsymbol{\omega})$ is defined as
$$\mathcal{L}_E(\boldsymbol{\omega}) := \mathbb{E}_{s_t}\Big[\big(\text{\normalfont Tr}[O^{\mathbf{z}}_t \rho^{\boldsymbol{\omega}_t}_t(\mathbf{z}|s_t))] - \exp{(-\lambda R^\pi_t(\boldsymbol{s}_T))}\big)^2\Big].$$

After updating the value parameters $\boldsymbol{\omega}$, we use them to build estimates of the value function and then update the policy parameters $\boldsymbol{\phi}$. Using the value estimates, we update the policy to minimize the loss, adapted from~\cite{murray_deep_2022}, defined as
\begin{align*}
    \mathcal{L}(\boldsymbol{\phi}) &:= \mathbb{E}_{s_t}\Big[ \frac{1}{\lambda} \exp\big(-\lambda (r_t(s_t,\pi^{\boldsymbol{\phi}_t}_t(s_t)) + v_{t+1}^{\boldsymbol{\omega}_{t+1}}(\boldsymbol{s}_{t+1}))\big) \big | s_t \Big] \\ & := \mathbb{E}_{s_t}\Big[ \frac{1}{\lambda} \exp\big(-\lambda r_t(s_t,\pi^{\boldsymbol{\phi}_t}_t(s_t)) \times \text{\normalfont Tr}[O^{\mathbf{z}}_{t+1}\rho^{\boldsymbol{\omega}_{t+1}}_t(\mathbf{z}|\boldsymbol{s}_{t+1}) ] \big| s_t \Big].
\end{align*}

\begin{algorithm}[t!]
	\caption{Expected Actor-Critic Deep Hedging with Compound Neural Networks}
	\label{algorithm:expected-deep-hedging}
	\begin{algorithmic}
		\STATE {\bfseries input} Policy QNN $\pi$, Value QNN $v$.
		\STATE {\bfseries hyperparameters} Number of episodes per training step $N$.
		\STATE Initialize policy and value QNNs with parameters $\{\boldsymbol{\phi}_{t}\}_{t=0}^{T}, \{\boldsymbol{\omega}_{t}\}_{t=0}^{T}$.
		\WHILE{True}
		\FOR{episode $i=1$ {\bfseries to} $N$}
		\FOR{time-step $t=0$ {\bfseries to} $T$}
		\STATE Compute action $a^i_t:=\text{\normalfont Tr}[O^{\mathbf{a}}_t\rho^{\boldsymbol{\phi}_{t}}_t(\mathbf{a}|s_t^i) ]$. 
		\STATE Take action $a^i_t$ and receive total reward $r^i_t := r_t^+ - r_t^-$.
		\ENDFOR
		\FOR{time-step $t=T$ {\bfseries to} $0$}
		\STATE Compute total cumulative return: 
		$\widetilde{R}_t^i = \sum_{t'=t}^T r^i_{t'}$.
		\ENDFOR
		\ENDFOR
		\STATE Update value parameters $\boldsymbol{\omega}$ with gradient descent to minimize: 
		$$\widetilde{\mathcal{L}}(\boldsymbol{\omega}) =  \frac{1}{N}\sum_{i=1}^{N}\sum_{t=0}^{T} \big( \text{\normalfont Tr}[O^{\mathbf{z}}_t \rho^{\boldsymbol{\omega}_t}_t(\mathbf{z}|s^i_t))] - \exp(-\lambda\widetilde{R}_t^i) \big)^2.$$
		\STATE Update policy parameters $\boldsymbol{\phi}$ with gradient descent to minimize:
		$$\widetilde{\mathcal{L}}(\boldsymbol{\phi}) =  \frac{1}{N}\sum_{i=1}^{N}\sum_{t=0}^{T}\frac{1}{\lambda}\exp( -\lambda r_t^i) \times \text{\normalfont Tr}[O^{\mathbf{z}}_{t+1} \rho^{\boldsymbol{\omega}_{t+1}}_t(\mathbf{z}|s^i_{t+1}))] \big). $$
		\ENDWHILE
		\STATE {\bfseries output} Policy parameters $\boldsymbol{\phi}$.
	\end{algorithmic}
\end{algorithm}

The training procedure for our approach is outlined in Algorithms~\ref{algorithm:expected-deep-hedging}~and~\ref{algorithm:distributional-deep-hedging}. At each iteration, we generate $N$ trajectories $\{s_t^i\}_{t=0}^T$ and use the policy QNN to compute the corresponding sequence of actions. Using this sequence of actions, we can compute the cumulative return for each episode and for each time-step and we use them to update the value network. When using the expected loss, we update the value parameters $\boldsymbol{\omega}$ such that we predict this cumulative return in expectation. In the distributional case, we need to compute the Hamming-weight of the future trajectory and update the parameters such that we predict the expectation for only that subspace. Once the value estimates are updated, we can use them to update the policy parameters $\boldsymbol{\phi}$.

\subsubsection{Properties of Quantum Deep Hedging}

Predicting the performance of any particular deep learning algorithm is difficult due to the complexity of the models used, and the often unpredictable behavior of the non-convex optimization that must be performed for training. It is however possible to examine some desirable global properties of the system that indicate (but do not guarantee) good performance. The structure of the presented Quantum Deep Hedging framework leads to some of these global properties, when specifically instantiated with Hamming-weight preserving unitaries (as in Section~\ref{sec:distributional-value}). Some of these properties have been hinted at throughout the text, and we summarize them here:
\newpage
\begin{algorithm}[H]
	\caption{Distributional Actor-Critic Deep Hedging  with Compound Neural Networks}
	\label{algorithm:distributional-deep-hedging}
	\begin{algorithmic}
		\STATE {\bfseries input} Policy QNN $\pi$, Value QNN $v$.
		\STATE {\bfseries hyperparameters} Number of episodes per training step $N$.
		\STATE Initialize policy and value QNNs with parameters $\{\boldsymbol{\phi}_{t}\}_{t=0}^{T}, \{\boldsymbol{\omega}_{t}\}_{t=0}^{T}$.
		\WHILE{True}
		\FOR{episode $i=1$ {\bfseries to} $N$}
		\FOR{time-step $t=0$ {\bfseries to} $T$}
		\STATE Compute action $a^i_t:=\text{\normalfont Tr}[O^{\mathbf{a}}_t\rho^{\boldsymbol{\phi}_t}_t(\mathbf{a}|s_t^i) ]$. 
		\STATE Take action $a^i_t$ and receive total reward $r^i_t := r_t^+ - r_t^-$.
		\ENDFOR
		\FOR{time-step $t=T$ {\bfseries to} $0$}
		\STATE Compute total cumulative return: 
		$\widetilde{R}_t^i = \sum_{t'=t}^T r^i_{t'}$.
		\STATE Compute Hamming-weight $k^{i}_t:= |s_T^i|-|s_t^i|$.
		\ENDFOR
		\ENDFOR
		\STATE Update value parameters $\boldsymbol{\omega}$ with gradient descent to minimize
		$$\widetilde{\mathcal{L}}(\boldsymbol{\omega}) =  \frac{1}{N}\sum_{i=1}^{N}\sum_{t=0}^{T} \big( \text{\normalfont Tr}[O^{\mathbf{z}}_t \rho^{\boldsymbol{\omega}_t}_t(\mathbf{z}|s^i_t,k^{i}_t+1))] - \exp(-\lambda\widetilde{R}_t^i) \big)^2.$$
		\STATE Update policy parameters $\boldsymbol{\phi}$ with gradient descent to minimize
		$$\widetilde{\mathcal{L}}(\boldsymbol{\phi}) =  \frac{1}{N}\sum_{i=1}^{N}\sum_{t=0}^{T}\frac{1}{\lambda}\exp( -\lambda r_t^i) \times \text{\normalfont Tr}[O^{\mathbf{z}}_{t+1}\rho^{\boldsymbol{\omega}_{t+1}}_t(\mathbf{z}|s^i_{t+1}) ] \big). $$
		\ENDWHILE
		\STATE {\bfseries output} Policy parameters $\boldsymbol{\phi}$.
	\end{algorithmic}
\end{algorithm}
\begin{itemize}
    \item \textbf{Expressivity:} The central question for any learning algorithm is whether the parameterized models used are expressive enough to capture target models of interest. The ``universality'' of models such as deep neural networks has been a driving force in their adoption and utility. In our algorithms we do not use models that are universal in that they can express any quantum operation, however we show that they are expressive enough to capture the quantities of interest, which in our case is the true distribution of the value function. The primary challenge is that the distribution of the value function is on an unknown and potentially changing support. We show in Proposition~\ref{prop:cramer} that our model that uses a fixed support and general parameterized unitaries on $m$-qubits can approximate the true distribution with error decaying exponentially in $m$. In Proposition~\ref{prop:cramer-bis}, we specialize this result to the case where the value function distribution is constant on the Hamming-weight subspaces and we correspondingly use a fixed support with Hamming-weight preserving parameterized unitaries. 

    \item \textbf{Generalization:} The number of possible futures in a deep-hedging environment grows exponentially with the time horizon $T$. In practice, our learning algorithm can only use a limited number of episodes (polynomial in $T$). We must therefore investigate the out-of-sample performance or generalization of our algorithm in this setting. This can however be guaranteed in our setting where we use the Quantum Compound Neural Networks. Our parameterized models consist of $\mathcal{O}(T)$ such networks, each on $\mathcal{O}(T)$ qubits. From the definition in Section~\ref{subsec:quantum-compound-nns}, each of the neural networks has $\mathcal{O}(T^2)$ parameters. As a consequence of results due to Caro and Datta~\cite{caro2020pseudo}, the pseudo-dimension of our parameterized model is polynomial in $T$. Therefore $\text{poly}(T)$ episodes suffice to ensure that empirical risk minimization over our sample converges with high probability to the true optimal expressible model over the whole distribution of futures.

    \item \textbf{Trainability:} Finally we consider whether the task of optimizing the parameters for our model can be performed efficiently. The associated optimization problem is non-convex and thus training convergence cannot be guaranteed. We can show however that in our setting, the well-known ``barren plateau'' problem~\cite{mcclean_barren_2018} does not arise. Each Quantum Compound Neural Network that we use is on $\mathcal{O}(T)$ qubits and has $\mathcal{O}(T)$ depth. Furthermore the loss function we measure can be constructed as a function of measurements in the computational basis (corresponding to a measuring the vector of observables $Z$ on each qubit $i$). We initialize the parameters of the model as normal random variables with variance $\mathcal{O}(1/T)$. Theorem~\ref{thm:gaussian_init_compound} ensures that the gradients decay only polynomially with the time horizon $T$.
\end{itemize}

\section{Results}
\label{sec:applications}

In the previous sections we introduced quantum methods for Deep Hedging, which use quantum orthogonal and compound neural networks within policy-search and actor-critic based reinforcement learning algorithms. In this section we present results of hardware experiments evaluating our methods on classical and quantum-accesible market environments.

 We benchmarked our models for both classical and quantum environments using three different methods: simulating our quantum models on classical hardware assuming perfect quantum operations, simulating them on classical emulators that model the noise for quantum hardware, and applying our quantum models directly on the $20$ qubit trapped-ion quantum processors Quantinuum H1-1, H1-2~\cite{h1}. Note that because orthogonal layers are efficiently simulatable classically, we can perform simulations for up to $64$ qubits, while for the compound architectures that use the entire exponential space , we only simulated layers with up to 12 qubits. 

\begin{table}[t!]
\centering
\begin{tabular}{lccc}
\toprule[1.5pt]
\multirow{2}{*}{Model} & \multicolumn{2}{c}{Utility} & Number of \\
& Without costs & With costs & parameters \\ \midrule[1.5pt]
Feed-forward (Classical) & $-2.868$ & $-5.064$ & $881$ \\
Feed-forward (Pyramid) & $-2.873$ & $-5.048$ & $521$ \\
Feed-forward (Butterfly) & $-2.874$ & $-5.043$ & $257$ \\
\midrule
Recurrent (Classical) & $-2.933$ & $-5.075$ & $881$ \\
Recurrent (Pyramid) & $-2.939$ & $-5.102$ & $521$ \\
Recurrent (Butterfly) & $-2.931$ & $-4.854$ & $257$ \\
\midrule
LSTM (Classical) & $-2.853$ & $-4.743$ & $569$ \\
LSTM (Pyramid) & $-2.856$ & $-4.755$ & $457$ \\
LSTM (Butterfly) & $-2.879$ & $-4.787$ & $217$ \\
\midrule
Transformer (Classical) & $-2.865$ & $-4.713$ & $1905$ \\
Transformer (Pyramid) & $-2.876$ & $-4.806$ & $1305$ \\
Transformer (Butterfly) & $-2.861$ & $-4.822$ & $865$ \\
\bottomrule[1.5pt]
\end{tabular}
\caption{Comparison of expected utilities without and with transaction costs for models with classical and orthogonal layers using exact simulation over 256 paths and 30 trading days, including the number of trainable parameters.}
\label{tab:part-1-classical-sim}
\end{table}

 In all experiments, the parameters of all the quantum compound neural networks were initialized using Gaussian initialization, and the training  for all quantum neural networks was performed in exact classical simulation. In the following subsections, we give the details of the results. 
        
\subsection{Classical Market Environment}
\label{subsec:results-classical-environment}

In the first part of our experiments, we consider Quantum Deep Hedging as described in Section~\ref{subsec:quantum-deep-hedging-classical-environment} in classical market environments. We considered the environment from~\cite{buehler_deep_2019,buehler_deep_2019-1}  where the authors used Black-Scholes model to simulate the market state and evaluate hedging strategies. In this setup, the underlying asset is modeled using Geometric Brownian Motion (GBM), which is commonly used in finance to model stock prices. 

A GBM is a continuous-time stochastic process $\boldsymbol{B}_t$ described by a stochastic differential equation
$$ d\boldsymbol{B}_t = \mu \boldsymbol{B}_t\,dt + \sigma \boldsymbol{B}_t\,d\boldsymbol{W}_t, $$
where $\mu \in \mathbb{R}$ is the percentage drift and $\sigma \in \mathbb{R^+}$ is the percentage volatility.  $\boldsymbol{W}_t$ corresponds to Brownian motion and thus $d\boldsymbol{W}_t \sim \mathcal{N}(0, dt)$.

For simulations, we assume one calendar year to be one unit of time increment and therefore set $dt=1/252$ assuming 252 trading days in a calendar year. The market state $s_t$ is represented by the sequence of past and actual market observations $\{M_{t'}\}_{t'=0}^{t}$, where $M_{t'}$ is the stock price at time-step $t'$. We used a European short call option with a strike price of $K=S_0$  as the instrument to be hedged. The time horizon was set to 30 trading days with daily rebalancing, and the percentage drift ($\mu$) for the GBM was set to 0 and the percentage volatility ($\sigma$) was set to 0.2. Proportional transaction costs were utilized with a proportionality constant of 0.01. The training dataset comprised of $9.6 \times 10^4$ samples, whereas the testing dataset consisted of $2.4 \times 10^4$ samples. We compared Feed-forward, Recurrent, LSTM, and Transformer models, 
 constructed using the framework described in Section~\ref{subsec:quantum-deep-hedging-classical-environment}. Here, the input sequence $(M_0, M_1, \dots, M_T) \in \mathbb{R}^{(t+1)}$ corresponds to the mid-market price of the underlying equity. 
The outputs $a_t^\pi\in\left[0,1\right]$ correspond to the model's delta for that time-step. The Feed-forward model is constructed using the \emph{Feed Forward} architecture. Both Recurrent and LSTM models are built using the \emph{Recurrent} architectures, where  in the Recurrent model the hidden state passed onto the subsequent time-step is fixed to be the output of the previous time-step, i.e. the model's position on the hedging instrument at the previous time-step. The Transformer model we used in this work is constructed by adding the attention mechanism on top of the \emph{Feed Forward} architecture. 

\begin{table}[t!]
\centering
\begin{tabular}{lccc}
\toprule[1.5pt]
\multirow{2}{*}{Model} & \multicolumn{2}{c}{Utility} & Number of \\
& Simulator & Emulator & circuits \\ \midrule[1.5pt]
Feed-forward (Butterfly) & $-5.041$ & $-5.155$ & $960$ \\
Recurrent (Butterfly) & $-5.006$ & $-5.333$ & $960$ \\
LSTM (Butterfly) & $-4.809$ & $-4.866$ & $3840$ \\
Transformer (Butterfly) & $-4.846$ & $-5.176$ & $2880$ \\
\bottomrule[1.5pt]
\end{tabular}
\caption{Comparison of exact simulation and Quantinuum H1-1 emulator results for orthogonal layer models, evaluating expected utilities with transaction costs over 32 paths and 30 trading days, and showing the number of circuits emulated.}
\label{tab:hard-emul-part-1}
\end{table}

 \paragraph{Exact Simulations.} To evaluate the behavior of quantum orthogonal neural networks, all four architectures (Feed-forward, Recurrent, LSTM, Transformer) were compared, both with classical linear and quantum orthogonal layers (with Pyramid and Butterfly circuits). A feature size of $16$ was used for the linear layers in classical architectures, and $16$ qubits were used for the orthogonal layers in quantum architectures. For Feed-forward and Recurrent models each hidden layer was repeated three times within the network. The LSTM model had one hidden cell constructed using four classical linear/quantum orthogonal layers. The Transformer model had three hidden layers followed by two classical linear/quantum orthogonal layers for the attention mechanism. Parameters for all models were shared across time-steps. Noiseless classical simulations were performed for training and inference, with a batch of $256$ paths. The results are presented in Table~\ref{tab:part-1-classical-sim}. We compared the achieved utilities with and without transaction costs and the number of training parameters. We observe that quantum orthogonal neural networks (Pyramid and Butterfly) achieve performance competitive with classical neural networks while using fewer trainable parameters and this holds for environments both with and without transaction costs. For the comparison, we have used the same classical and quantum architectures with the same layer sizes and we trained with identical hyperparameters. The quantum networks showed competitive performance and used fewer parameters due to the fact that every linear layer had been replaced with an orthogonal one. Let us also note that it might be possible to achieve  parameter reduction with other classical methods (e.g. pruning). The Transformer and LSTM architectures demonstrated the highest model utilities among the studied architectures, while the quantum orthogonal Butterfly layers used fewest training parameters.

\paragraph{Hardware Emulations.}
To investigate the behavior of our quantum neural networks on current hardware, we employed Quantinuum H1-1 emulator~\cite{h1} to perform inference on our models. We kept the same environment configuration and used a batch of $32$ paths to perform inference. However, we downsized the network to a single layer as this allowed fewer circuit executions without significantly hampering model utilities. The Feed-forward and Recurrent architectures use one circuit evaluation per time-step, while the LSTM architectures use $4$ circuit evaluations per time-step, and the Transformer architectures use $3$ circuit evaluations per time-step. As the hardware architecture allowed for all-to-all connectivity, we used quantum orthogonal layers with a Butterfly circuit which enables log-depth circuits with linear number of two-qubit gates and thus are ideal for computations on near-term hardware. We used 1000 measurement shots per circuit evaluation to perform tomography over the unary basis and construct the output of each layer. The results are summarized in Table~\ref{tab:hard-emul-part-1}. The utility of the models is presented for two cases: when evaluated on a classical exact simulator and when evaluated on Quantinuum's hardware emulator. The table also summarizes the number of circuit evaluations needed to hedge 32 paths over 30 days with each model architecture. The results show that the LSTM architecture with Butterfly layers were most robust to noise as the model utilities on the simulator and emulator are relatively close. We also observed that the LSTM model achieved the highest utility on both cases. 

\begin{table}[t!]
\centering
\begin{tabular}{lcccccc}
\toprule[1.5pt]
\multirow{2}{*}{Model}   & \multirow{2}{*}{Utility} & \multicolumn{4}{c}{Terminal PnLs} \\ 
&  & Path 1 & Path 2 & Path 3 & Path 4 \\ \midrule[1.5pt]
LSTM (Classical)                     & $-2.173$  & $-2.578$ & $-1.225$ & $-1.420$ & $-2.671$  \\ 
LSTM (Butterfly – Simulation)        & $-2.176$  & $-2.586$ & $-1.194$ & $-1.439$ & $-2.671$  \\
LSTM (Butterfly – Hardware)          & $-2.194$  & $-2.610$ & $-1.284$ & $-1.488$ & $-2.658$  \\ 
\midrule
Transformer (Classical)              & $-2.167$  & $-2.563$ & $-1.219$ & $-1.411$ & $-2.673$  \\ 
Transformer (Butterfly – Simulation) & $-2.195$  & $-2.639$ & $-1.242$ & $-1.388$ & $-2.672$  \\
Transformer (Butterfly – Hardware)   & $-2.539$  & $-3.341$ & $-1.355$ & $-1.247$ & $-2.713$  \\ 
\bottomrule[1.5pt]
\end{tabular}
\caption{Comparison of exact simulation and Quantinuum H1-1 hardware results for orthogonal layer models, evaluating expected utilities and terminal PnLs with transaction costs over $4$ paths and $5$ trading days, evaluating the performance under hardware conditions.}
\label{tab:hard-exp-part-1}
\end{table}

\paragraph{Hardware Experiments.}

For our hardware experiments, we used the LSTM and Transformer models with 16-qubit Butterfly quantum circuits to perform inference on the Quantinuum H1-1 trapped-ion quantum processor~\cite{h1}. We reduced the time horizon of the GBM to 5 days and considered models with transaction costs for a batch of 4 randomly chosen paths. We used the same model size as the ones used on hardware emulators which resulted in 80 circuit executions for the LSTM model and 60 circuit executions for the Transformer model. We present inference results for a model with a classical linear layer, and a quantum model with a butterfly quantum orthogonal layer simulated and executed on the quantum hardware. The results are 
presented in Table~\ref{tab:hard-exp-part-1}. In addition to model utilities, we also list the terminal Profit and Loss (PnL) for each path for a more fine-grained comparison.
The results reveal that the LSTM architecture exhibits robustness to noise, consistent with the results obtained from the hardware emulator, as evidenced by the terminal PnL values of each path closely aligning with those of the simulations run on the hardware. Conversely, the Transformer's hardware execution demonstrates poorer performance compared to the simulation results.

\subsection{Quantum Market Environment}
\label{subsec:results-quantum-environment}

In the second part of our experiments, we utilize quantum compound neural networks in various reinforcement learning algorithms in a quantum environment. Specifically, we implement the expected and distributional actor-critic algorithms, as described in Section~\ref{subsec:quantum-deep-hedging-part-2}, and compare them to the policy-search algorithm adapted for compound neural networks. To accomplish this, we first describe how to construct a quantum environment for Quantum Deep Hedging by adapting the classical environment used in Section~\ref{subsec:results-classical-environment}. More specifically, we aim to build a quantum environment that mimics market dynamics following the Black-Scholes model, as described by a GBM.

To encode the dynamics of the Brownian motion in a quantum environment, we can use the fact that Brownian motions can be seen as the limit of a discrete random walk. Specifically, we can use a sequence of $nT$ independent and identically distributed Bernoulli random variables $\boldsymbol{b}_1,\boldsymbol{b}_2,\dots, \boldsymbol{b}_{nT}$ with mean $1/2$ to approximate $\boldsymbol{W}_T$, where $T$ is the maturity and $n$ is a hyperparameter that determines the precision of the approximation. Using this property, we can provide a discrete quantum environment for the Black-Scholes model, and approximate the price $\boldsymbol{B}_t$ at some time-step $t$ by $\boldsymbol{b}_1, \boldsymbol{b}_2, \dots, \boldsymbol{b}_{nt}$ as follows:
$$\boldsymbol{B}_t \approx  B_0 \times \exp\Big((\mu-\frac{\sigma^2}{2})t + \frac{\sigma}{\sqrt{n}}\sum_{k=1}^{nt} (2\boldsymbol{b}_k-1)\Big).
$$

\begin{figure}[t!]
\centering
\includegraphics[width=0.9\textwidth]{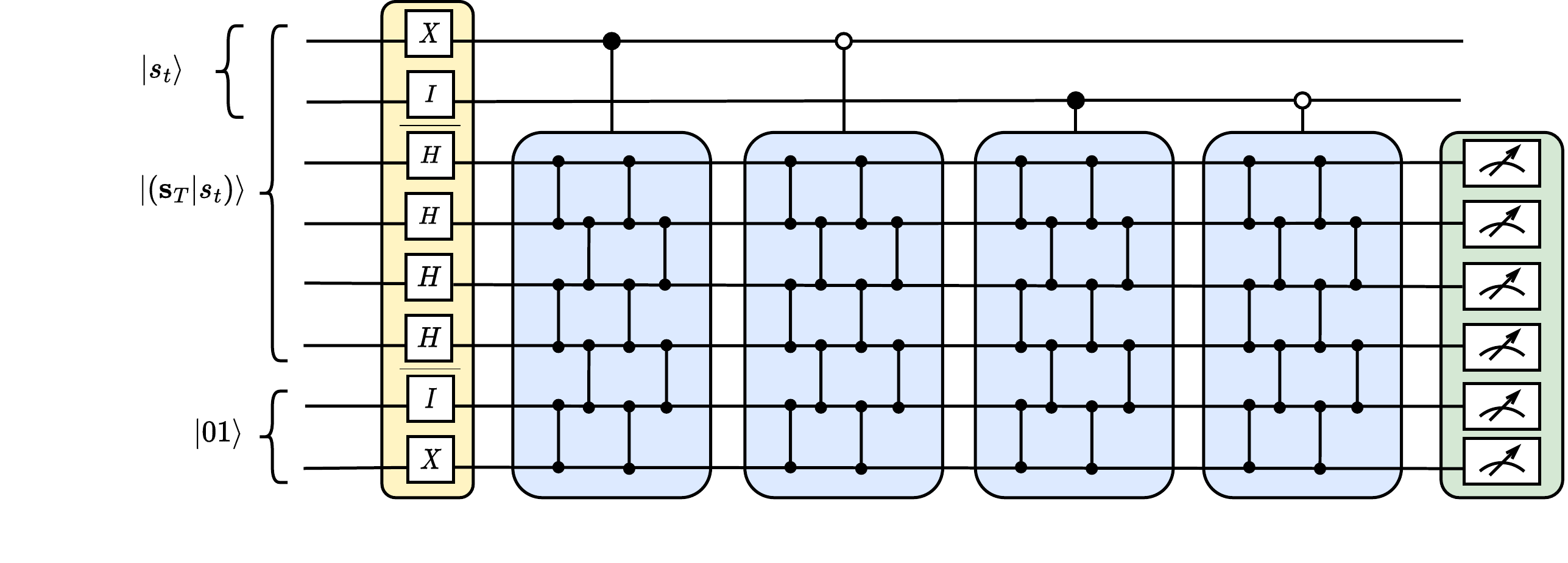}
\caption{The quantum compound neural network using $\mathcal{O}(T)+2$ qubits for the Black-Scholes model. For loading the data (in yellow), first $t$ qubits are used to encode past jumps of the market state. The next $\mathcal{O}(T-t)$ qubits encode the transition oracles which in the case of Black-Scholes corresponds to an equal superposition over all possible future jumps. The last two qubits are ancilla and are used to encode the state $|01\rangle$. Next for the unitaries (in blue), we used a architecture controlled on past market state. Based on the direction of jump, a different compound layer constructed using Brick architecture (Figure~\ref{fig:orthogonal-layers}) is applied on $T-t+2$ qubits. Finally, each of the last $T-t+2$ qubits is measured (in green) independently. For $t=0$, the input state is an equal superposition of all possible future jumps followed by one unitary over $T+2$ qubits without control. As described in Section~\ref{subsec:compound-circuit-construction} the control qubits are always in a computational basis state. Thus, they can be removed for efficient hardware implementation, and the appropriate parameters for unitary acting on $T-t+2$ qubits can be chosen classically for each past state.}
\label{ansatz-controlled}
\end{figure}

To obtain a sample $B_t$ of $\boldsymbol{B}_t$, we sample $nt$ Bernoulli variables $b_1, \dots, b_{nt}$, i.e., $n$ Bernoulli variables per day. Thus, we define the encoding of the market observation $M_t$ for a time-step $t>0$ as $\ket{M_t}:=\ket{b_{n(t-1)+1}\dots b_{nt}}$, which contains all the jumps between $t$ and $t+1$ and that can be encoded using $n$ qubits. We define the encoding of the market state $s_t$ at time-step $t$ as the history of all previous jumps, i.e., $\ket{s_t}:=\ket{b_1b_2\dots b_{nt}}$, from which we retrieve the price. Loading this quantum state can be done using a $1$-depth circuit made with at most $nt$ Pauli-$X$ gates acting on $nt$ qubits. Note that the number of qubits required to encode $s_t$ here is $nt$ and not $n(t+1)$ as in the general case, since the price at time-step $0$ is fixed. For every time-step $t$, the transition model $p_t$ can be oracularized by applying $n$ Hadamard gates on an additional $n$ qubits. The different transition oracles can be applied in parallel to build a superposition of all future trajectories and obtain$$\ket{(\boldsymbol{s}_T|s_t)} = \ket{s_t}\otimes\sum_{b\in\{0,1\}^{n(T-t)}} \frac{1}{2^{n(T-t)/2}}\ket{b}.$$

In our experimental setup, we chose $n=1$ as the number of Bernoulli variables per day to approximate Brownian motion. We retained the same GBM parameters as in the classical environment ($\mu=0$ and $\sigma=0.2$), but set the instrument maturity to $10$ days. Due to the limitations of simulating circuits with up to $12$ qubits for compound architectures, we adjusted the time step increment to $30$ trading days, instead of the usual $252$. This means each time-step $t$ should be changed to $t/30$ for the approximate GBM simulations, allowing us to capture short-term stock price fluctuations while maintaining the overall GBM behavior. The final payoff is a European short call option with a strike price of $k=1$. We also investigated cases with and without a transaction cost proportional to $\epsilon=0.002$.

\begin{table}[htp!]
\centering
\begin{tabular}{lcc}
\toprule[1.5pt]
\multirow{2}{*}{Algorithm} &
\multicolumn{2}{c}{Utility} \\
& Without costs & With costs\\ \midrule[1.5pt]
Policy-search  & $-4.064$ & $-4.639$\\
Expected actor-critic & $-4.193$ & $-4.668$ \\
Distributional actor-critic & $-3.875$ & $-4.424$ \\
\bottomrule[1.5pt]
\end{tabular}
\caption{Comparison of compound neural networks trained using different algorithms with exact simulation, evaluating expected utilities over $16$ paths and $10$ trading days, both without and with transaction costs.}
\label{tab:part-2-classical-sim}
\end{table}

\newpage
\begin{table}[htp!]
    \centering
    \begin{tabular}{lccc}
    \toprule[1.5pt]
    \multirow{2}{*}{Algorithm} & \multicolumn{2}{c}{Utility} & Number of \\
    & Simulator & Emulator & circuits \\ \midrule[1.5pt]
    Policy-search  & $-4.257$ & $-4.277$ & $160$ \\
    Expected actor-critic & $-4.528$ & $-4.531$ & $160$ \\ 
    Distributional actor-critic & $-4.185$ & $-4.180$ & $160$ \\
    \bottomrule[1.5pt]
    \end{tabular}
    \caption{Comparison of compound neural networks trained using different algorithms with simulation and Quantinuum H1-1 emulator results, evaluating expected utilities with transaction costs over 16 paths and 10 trading days.}
    \label{tab:hard-emul-10-days}
\end{table}

We utilized the compound neural networks from Section~\ref{subsec:compound-circuit-construction} to represent the policy in the policy-search algorithm and both the policy and value in the actor-critic algorithms. We employed the Huber loss and scaled the value function to prevent exploding gradients. The algorithms were trained using classical simulations of the quantum circuits for $2000$ steps with Adam optimizers, employing $3$ random seeds and a batch of $16$ generated episodes per training step. We selected the best parameters from these runs for a random selection of $16$ paths and reported the inference results.

\paragraph{Exact Simulations.} Here, the performance of the algorithms was evaluated through exact simulation on classical hardware using the Brick architecture with logarithmic depth per block for training the compound neural networks. The results, presented in Table~\ref{tab:part-2-classical-sim}, showed that policies trained using a distributional actor-critic algorithm yielded better utilities for this particular example. These results align with the findings by Lyle et al.~\cite{lyle_comparative_2019}, where minimizing a distributional loss led to better policies.

\paragraph{Hardware Emulations.} We investigated the performance of the algorithms in presence of hardware noise by running inference on the Quantinuum H1-1 emulator~\cite{h1} over 16 randomly chosen paths. To accommodate today's hardware limitations, the depth of circuits with large depth was reduced by using a fixed depth per block instead of logarithmic depth per block. We compared the inference results of the hardware emulator with exact simulation. Results presented in Table~\ref{tab:hard-emul-10-days} show that quantum compound neural networks are noise-resilient, with similar utilities demonstrated between classical simulation and hardware emulation. Furthermore, our results show that the distributional policies outperformed the expected policies, and the expected policies outperformed the policies trained using the policy-search Deep Hedging algorithm.

\begin{table}[ht]
    \centering
    \resizebox{\linewidth}{!}{
    \begin{tabular}{lcccccccccc}
    \toprule[1.5pt]
    \multirow{2}{*}{Algorithm}   & \multirow{2}{*}{Utility} & \multicolumn{8}{c}{Terminal PnLs} \\ 
    &  & Path 1 & Path 2 & Path 3 & Path 4 & Path 5 & Path 6 & Path 7 & Path 8 \\ \midrule[1.5pt]
    Black Scholes & $-4.884$ & $-4.602$ & $-5.373$ & $-4.614$ & $-4.263$ & $-5.173$ & $-5.030$ & $-5.017$  & $-4.962$ \\ \midrule
    Expected actor-critic (Simulation) & $-3.547$ & $+0.078$   & $-6.204$  & $-0.203$  & $+0.967$   & $-6.768$  & $-3.071$  & $-2.984$   & $-6.689$  \\
    Expected actor-critic (Hardware) & $-3.501$  & $+0.213$  & $-6.666$  & $-0.556$  & $+1.067$   & $-6.895$  & $-2.315$  & $-2.569$   & $-6.556$ \\ \midrule
    Distributional actor-critic (Simulation) & $-3.309$  & $-1.807$  & $-8.313$  & $-3.803$ & $+1.464$   & $-2.736$  & $-1.934$  & $-2.669$   & $-3.944$  \\
    Distributional actor-critic (Hardware)  & $-3.369$  & $-1.802$  & $-8.214$  & $-3.648$  & $+1.367$   & $-2.993$  & $-2.047$  & $-2.803$   & $-4.200$  \\ \bottomrule[1.5pt]
    \end{tabular}}
    \caption{Comparison of compound neural networks trained using different algorithms with exact simulation and Quantinuum H1-1, H1-2 hardware results, evaluating expected utilities and terminal PnLs with transaction costs over $8$ paths and 10 trading days, benchmarked against the standard Black-Scholes delta-hedging model.}
    \label{tab:hard-exp-part-2}
\end{table}

\paragraph{Hardware Experiments.} In the third part of our experiments, we performed inference on Quantinuum's H1-1 and H1-2 trapped-ion quantum processors~\cite{h1} using policies trained via distributional and expected algorithms. We used a set of 8 randomly chosen paths and compared the terminal PnLs and utility in presence of transaction costs obtained via quantum hardware with exact classical simulations. The results are presented in Table~\ref{tab:hard-exp-part-2}. We also present the results from Black-Scholes delta hedge model for the setting . We note that the utility obtained from the hardware closely aligns with the emulation results, and the PnL values for the selected paths are also similar. Our study reveals that both the distributional and expected policies significantly outperformed the Black-Scholes delta hedge, with the distributional policy exhibiting the best overall performance.

\section{Discussion}

In this work we developed quantum reinforcement learning methods for Deep Hedging. These methods are based on novel quantum neural network architectures that utilize orthogonal and compound layers, and on a novel distributional actor-critic algorithm that takes advantage of the fact that quantum states and operations naturally deal with large distributions. 

There are many potential advantages to using quantum methods to enhance the capabilities of Deep Hedging algorithms. First, for neural networks with deep architectures, as is the case for time-series data, feature orthogonality can improve interpretability, help to avoid vanishing gradients and  result in faster and better training. Second, quantum compound neural networks can explore a larger dimensional optimization landscape and thus might train to more accurate models, once we ensure that barren plateau phenomena do not occur. Third, quantum neural networks are natively appropriate to be used in distributional reinforcement learning algorithms, which can lead to considerably better models, as we show for the toy example developed in this work. Finally, the quantum circuits that we need to implement in order to train competitive quantum models for Deep Hedging are rather small, since the number of qubits and depth of the quantum circuit is basically equal to the maturity time.

Note that our hardware experiments were done with a maturity of $10$ days, due to the fact that we had to simulate the training of the quantum models on a classical computer which very soon becomes infeasible. In principle, one can train directly on the quantum computer, using for example a parameter shift rule to compute the gradients~\cite{schuld_evaluating_2019}, in which case even with the current state of the quantum hardware (or with the hardware that will arrive in the next years) one can indeed train a Quantum Deep Hedging model for a maturity time of a month or more. In this case, the quantum model can no longer be simulated classically.

Moreover, we believe our quantum reinforcement learning methods have applications beyond Deep Hedging, for example for algorithmic trading or option pricing, and it would be interesting to develop specific quantum methods for such problems. Note that in these use cases the training data can be produced efficiently, removing the bottleneck of loading large amounts of data onto the quantum computer. 

The open questions regarding our work in quantum reinforcement learning are centered around three aspects. First, there is a need to expand the results regarding the trainability of the quantum neural networks proposed in this work to other settings. Second, the question of how to extend the quantum environment built for the GBM to other environments such as the Heston model studied in~\cite{murray_deep_2022} arises. Finally, there is a need to design new distributional losses that make use of temporal difference methods to learn the value functions in the Deep Hedging context. One approach to this is using the theoretical framework that allows for such design as developed in~\cite{buehler_deep_2022-1}, while another approach is using the moment matching approach as described in~\cite{nguyen-tang_distributional_2021}. Currently, the work focuses on the expectation of the value function, but it is important to consider other moments that can be matched for both the overall expectation and the expectation per subspace.
\label{sec:conclusions}

\section*{Code Availability}
The code for the numerical experiments in this paper is available upon request.

\section*{Acknowledgments}

We would like to thank Phillip Murray from JPMorgan Chase for his help with the actor-critic experiments. Special thanks are also due to Tony Uttley, Jenni Strabley, and Brian Neyenhuis from Quantinuum for their assistance on the execution of the experiments on the Quantinuum H1-1, H1-2 trapped-ion quantum processors. 

\noindent I.K. acknowledges support  by projects EPIQ (ANR-22-PETQ-0007), QUDATA (ANR-18-CE47-0010), QOTP (QuantERA ERA-NET Cofund 2022-25), HPCQS (EuroHPC 2021-24).

\section*{Disclaimer}
This paper was prepared for informational purposes with contributions from the Global Technology Applied Research center of JPMorgan Chase \& Co. This paper is not a product of the Research Department of JPMorgan Chase \& Co. or its affiliates. Neither JPMorgan Chase \& Co. nor any of its affiliates makes any explicit or implied representation or warranty and none of them accept any liability in connection with this paper, including, without limitation, with respect to the completeness, accuracy, or reliability of the information contained herein and the potential legal, compliance, tax, or accounting effects thereof. This document is not intended as investment research or investment advice, or as a recommendation, offer, or solicitation for the purchase or sale of any security, financial instrument, financial product or service, or to be used in any way for evaluating the merits of participating in any transaction.

\newpage
\bibliographystyle{quantum}
\bibliography{main}

\end{document}